%% file: main.tex
\documentclass[a4paper]{llncs}
\usepackage{wrapfig}
\usepackage{etex}
\usepackage[english]{babel} 
\usepackage{latexsym}
\usepackage{epsfig} 
\usepackage{listings}
\usepackage{url,xspace} 
\usepackage[draft]{fixme}
\usepackage{color} 
\usepackage{multirow}
\usepackage{latexsym,amssymb}
\usepackage[noend]{algpseudocode}
\usepackage{amsmath}
\usepackage{graphics}
\usepackage{graphicx}
\usepackage{todonotes}
\usepackage{pifont}

\usepackage[symbol]{footmisc}

\usepackage[a4paper,margin=1in,footskip=0.25in]{geometry}

\input defs.tex

\usepackage[all]{xy}

\usepackage{algorithm}
\usepackage{listings}
\begin{document}
\title{Analyzing Smart Contracts: From EVM to a sound Control-Flow
  Graph}
%
%\vspace{-15pt}\
%
\author{Elvira Albert$^{1,2}$ \and Jes\'us Correas$^2$ \and Pablo
  Gordillo$^2$ \and Alejandro Hern\'andez-Cerezo$^2$ \and \\Guillermo Rom\'an-D\'iez$^3$ \and Albert
  Rubio$^{1,2}$ }

\institute{
  Instituto de Tecnolog\'ia del Conocimiento, Spain\and
    Complutense University of Madrid,  Spain \and
  Universidad Polit\'ecnica de Madrid, Spain}
\maketitle

\input abstract

% \input{intro}

%\input{tool}

\input{evmlanguage}

\input{cfganalysis}

% \mycomment{PENDING: MERGE ORIGINAL TECHREP TEXT WITH LATEST VERSION OF
%   APPENDIX THAT FOLLOWS:}

% \input{cfganalysis-appendix}

% \input proof1

% \input proof2

\bibliographystyle{plain}
\bibliography{biblio}

\end{document}

%% file: defs.tex
\definecolor{deepblue}{rgb}{0,0,0.5}
\definecolor{deepred}{rgb}{0.6,0,0}
\definecolor{deepgreen}{rgb}{0,0.5,0}
\definecolor{halfgray}{gray}{0.55}
\definecolor{ipythonframe}{RGB}{207, 207, 207}

\definecolor{ckeyword}{HTML}{7F0055}
\definecolor{ccomment}{HTML}{3F7F5F}
\definecolor{cnumber}{HTML}{2A0099}
\definecolor{pblue}{rgb}{0.13,0.13,1}
\definecolor{pgreen}{rgb}{0,0.5,0}
\definecolor{pred}{rgb}{0.9,0,0}
\definecolor{pgrey}{rgb}{0.46,0.45,0.48}

\definecolor[named]{ACMBlue}{cmyk}{1,0.1,0,0.1}
\definecolor[named]{ACMYellow}{cmyk}{0,0.16,1,0}
\definecolor[named]{ACMOrange}{cmyk}{0,0.42,1,0.01}
\definecolor[named]{ACMRed}{cmyk}{0,0.90,0.86,0}
\definecolor[named]{ACMLightBlue}{cmyk}{0.49,0.01,0,0}
\definecolor[named]{ACMGreen}{cmyk}{0.20,0,1,0.19}
\definecolor[named]{ACMPurple}{cmyk}{0.55,1,0,0.15}
\definecolor[named]{ACMDarkBlue}{cmyk}{1,0.58,0,0.21}

\lstdefinelanguage{Solidity}{ keywords=[1]{anonymous, synchronized,
	assembly, assert, break, case, catch, class, constant,
	continue, contract, debugger, default, delete, do, else,
	event, export, external, finally, for, function, if,
	implements, import, in, indexed, instanceof, interface,
	internal, is, length, library, 
         modifier, new, payable,
	pragma, private, protected, public, pure, require, returns,
	send, 
        struct, super, switch, then, throw, transfer,
	try, typeof, using, view, while, with, ecrecover}, % generic
	keywordstyle=[1]\bfseries,
	keywords=[2]{bool, bytes, bytes1, bytes2, bytes3, bytes4, bytes5, bytes6, bytes7, bytes8, bytes9, bytes10, bytes11, bytes12, bytes13, bytes14, bytes15, bytes16, bytes17, bytes18, bytes19, bytes20, bytes21, bytes22, bytes23, bytes24, bytes25, bytes26, bytes27, bytes28, bytes29, bytes30, bytes31, bytes32, enum, int, int8, int16, int24, int32, int40, int48, int56, int64, int72, int80, int88, int96, int104, int112, int120, int128, int136, int144, int152, int160, int168, int176, int184, int192, int200, int208, int216, int224, int232, int240, int248, int256, mapping, string, uint, uint8, uint16, uint24, uint32, uint40, uint48, uint56, uint64, uint72, uint80, uint88, uint96, uint104, uint112, uint120, uint128, uint136, uint144, uint152, uint160, uint168, uint176, uint184, uint192, uint200, uint208, uint216, uint224, uint232, uint240, uint248, uint256, var, void, ether, finney, szabo, wei, days, hours, minutes, seconds, weeks, years},	% types; money and time units
	keywordstyle=[2]\bfseries,
	keywords=[3]{   %%%block,
        this, true, false, msg, data, sender, value, sig,
	value, now, tx},	% environment
    keywordstyle=[3]\bfseries,
    identifierstyle=, sensitive=false, comment=[l]{//},
    morecomment=[s]{/*}{*/}, commentstyle=\color{gray}\ttfamily,
    stringstyle=\ttfamily, morestring=[b]', morestring=[b]"
    }

\newcommand{\scode}[1]{\lstinline[language=Solidity,basicstyle=\small\ttfamily]{#1}}

\newcommand{\code}[1]{\scode{#1}}

\definecolor{shadecolor}{gray}{1.00}
\definecolor{ddarkgray}{gray}{0.75}
\definecolor{darkgray}{gray}{0.30}
\definecolor{light-gray}{gray}{0.87}

\newcommand{\ie}{\emph{i.e.}\xspace}

\newcommand{\eg}{\emph{e.g.}\xspace}

\newcommand{\set}[1]{\left\{{#1}\right\}}

\newcommand{\plname}[1]{\textsf{#1}\xspace}

\newcommand{\solidity}{\plname{Solidity}}

\newcommand{\EVM}{\textsf{EVM}\xspace}
\newcommand{\evm}{\textsf{EVM}\xspace}

\newcommand{\tuple}[1]{\langle #1 \rangle}

\newcommand{\rrderivproof}{\Rightarrow}

\makeatletter
\newsavebox\myboxA
\newsavebox\myboxB
\newlength\mylenA

\newcommand*\xoverline[2][0.75]{%
    \sbox{\myboxA}{$\m@th#2$}%
    \setbox\myboxB\null% Phantom box
    \ht\myboxB=\ht\myboxA%
    \dp\myboxB=\dp\myboxA%
    \wd\myboxB=#1\wd\myboxA% Scale phantom
    \sbox\myboxB{$\m@th\overline{\copy\myboxB}$}%  Overlined phantom
    \setlength\mylenA{\the\wd\myboxA}%   calc width diff
    \addtolength\mylenA{-\the\wd\myboxB}%
    \ifdim\wd\myboxB<\wd\myboxA%
       \rlap{\hskip 0.7\mylenA\usebox\myboxB}{\usebox\myboxA}%
    \else
        \hskip -0.5\mylenA\rlap{\usebox\myboxA}{\hskip 0.7\mylenA\usebox\myboxB}%
    \fi}
\makeatother

\lstdefinestyle{numbers}
{numbers=left, numberstyle=\tiny}

\newcommand{\ptabstract}{\pi}

\newcommand{\ptabstracttop}{\ptabstract_{_\top}}
\newcommand{\ptabstractbot}{\ptabstract_{_\bot}}
\newcommand{\eqn}{{\mathcal{X}}}
\newcommand{\eq}[1]{{\mathcal{X}_{\hdir{#1}}}}
\newcommand{\lub}{\sqcup}

\def\anno#1{{\ooalign{\hfil\raise.07ex\hbox{\small{\rm
          \textcolor{red}{{\scriptsize #1}}}}\hfil%
        \crcr{\small \textcolor{blue}{\mathhexbox20D}}}}}

\def\annoblack#1{{\ooalign{\hfil\raise.07ex\hbox{\small{\rm
          \textcolor{black}{{\scriptsize #1}}}}\hfil%
        \crcr{\small \textcolor{black}{\mathhexbox20D}}}}}

\newcommand{\fstate}[2]{\{\tuple{#1} \mapsto \{\tuple{#2}\}\}}
\newcommand{\hdir}[1]{\text{\texttt{#1}}}
\newcommand{\CFG}[0]{\ensuremath{\textit{CFG}}}

\newcommand{\rrrulewide}[2]{
\begin{minipage}{0.9\textwidth}
\vspace*{0.1cm}
\begin{center}
\ensuremath{
\begin{array}{c}
#1 \\
\hline
#2 
\end{array}
}
\end{center}
\end{minipage}
}

\newcommand{\ruleskip}{ \\ } % \hspace{20pt}}

\newcommand{\proofend}{\hfill {\tiny $\square$}}
\newcommand{\exampleend}{\hfill {\tiny $\blacksquare$}}
\newcommand{\mycomment}[1]{{\color{black}{#1}}}

\newcommand{\rulename}[2]{\ensuremath{rulename}}

%% file: abstract.tex
% ;; -*- coding: iso-latin-1; TeX-PDF-mode: t; TeX-master: "main" -*-%

\begin{abstract}

  The \EVM language is a simple stack-based language with words of 256
  bits, with
  one significant difference between the \EVM and other virtual
  machine languages (like Java Bytecode or CLI for .Net programs): the
  use of the stack for saving the jump addresses instead of having it
  explicit in the code of the jumping instructions.
  Static analyzers need the complete control flow graph (CFG) of the
  \EVM program in order to be able to represent all its execution
  paths.
  This report addresses the problem of obtaining a precise and
  complete stack-sensitive CFG by means of a static analysis, cloning
  the blocks that might be executed using different states of the
  execution stack. The soundness of the analysis presented is proved.

\end{abstract}

%% file: evmlanguage.tex
% ;; -*- coding: iso-latin-1; TeX-PDF-mode: t; TeX-master: "main" -*-%

\section{\EVM Language}
\label{sec:evm-language}

The \EVM language is a simple stack-based language with words of 256
bits with a local volatile memory that behaves as a simple
word-addressed array of bytes, and a persistent storage that is part
of the blockchain state.  A more detailed description of the language
and the complete set of operation codes can be found
in~\cite{yellow}. In this section, we focus only on the relevant
characteristics of the \EVM that are needed for describing our work.
\mycomment{We will consider \EVM programs that satisfy two
  constraints: (1) jump addresses are constants,\ie they are
  introduced by a \code{PUSH} operation, they do not depend on input
  values and they are not stored in memory nor storage, and (2) the
  size of the stack when executing a jump instruction can be bounded
  by a constant. These two cases are mostly produced by the use of
  recursion and higher-order programming in the high-level language
  that compiles to \EVM, as \eg Solidity.}

% Though \toolname
%receives an \EVM program as input, we describe the example with a high level
%language, Solidity, which can be compiled to \EVM by means of \texttt{solc}
%compiler.

%\input{instructionset}

\input{running-example}

\begin{example}

In order to describe our techniques, we use as running example a simplified
version (without calls to the external service Oraclize and the authenticity
proof verifier)
% and without a function call (see below),
of the % \texttt{\_\_callback} method extracted from the \code{EthereumPot}
contract~\cite{etherpot} that implements a lottery system.
% intermediate representations involved in our analysis we will use as
% running example the \texttt{\_\_callback} method extracted from the
% EthereumPot~\cite{etherpot} contract.  EthereumPot implements a
% simple lottery.
During a game, players %represented by arbitrary Ethereum
% accounts,
call a method \code{joinPot} to buy lottery tickets; each player's
address is appended to an array \code{addresses} of current players,
and the number of tickets is appended to an array \code{slots}, both
having variable length.  After some time has elapsed, anyone can call
\code{rewardWinner} which calls the \code{Oraclize} service to obtain
a random number for the winning ticket.  If all goes according to
plan, the \code{Oraclize} service then responds by calling the
\code{__callback} method with this random number and the authenticity
proof as arguments.
% (in the version we are analysing, the function call that verifies
% this proof is not included).
A new instance of the game is then started, and the winner is allowed
to withdraw her balance using a \code{withdraw}
method. Figure~\ref{fig:solevm} shows an excerpt of the \solidity code
(including the public function \code {findWinner}) and a fragment of
the \EVM code produced by the compiler. The \solidity source code is
shown for readability, as our analysis works directly on the \EVM
code.
%(if it receives the source, it first compiles it
%to obtain the \EVM code).

\input{running-example}
\noindent To the right of Figure~\ref{fig:solevm} we show a fragment
of the \EVM code of method \code{findWinner}. It can be seen that the
\EVM has instructions for operating with the stack contents, like
\code{DUP}\textit{x} or \code{SWAP}\textit{x}; for comparisons, like
\code{LT}, \code{GT}; for accessing the storage (memory) of the
contract, like \code{SSTORE}, \code{SLOAD} (\code{MLOAD},
\code{MSTORE}); to add/remove elements to/from the stack, like
\code{PUSH}\textit{x}/ \code{POP}; and many others (we again refer
to~\cite{yellow} for details).
% % %
Some instructions increment the program counter in several units
(e.g., \code{PUSH}\textit{x}~\code{Y} adds a word with the constant
\code{Y} of \textit{x} bytes to the stack and increments the program
counter by $x+1$). In what follows, we use $size(b)$ to refer to the
number of units that instruction $b$ increments the value of the
program counter. For instance $size(\text{\code{POP}}) = 1$,
$size(\text{\code{PUSH1}}) = 2$ or $size(\text{\code{PUSH3}}) = 4$.
\exampleend
\end{example}

One significant difference between the \EVM and other virtual machine
languages (like Java Bytecode or CLI for .Net programs) is the use of
the stack for saving the jump addresses instead of having it explicit
in the code of the jumping instructions.  In \EVM, instructions
\code{JUMP} and \code{JUMPI} will jump, unconditionally and
conditionally respectively, to the program counter stored in the top
of the execution stack. This feature of the \EVM requires, in order to
obtain the control flow graph of the program, to keep track of the
information stored in the stack. Let us illustrate it with an example.
% Figure~\ref{fig:instructionset} shows the subset of instructions that are
% relevant for \toolname.

\begin{example}
\label{ex:semantics}

In the \EVM code to the right of Figure~\ref{fig:solevm} we can see
two jump instructions at program points \hdir{953} and \hdir{660},
respectively, and the jump address (\hdir{64B} and \hdir{6D0}) is
stored in the instruction immediately before them: \hdir{950} or
\hdir{65D}.
% %
It then jumps to this destination by using the instruction
\hdir{JUMPDEST} (program points \hdir{941}, \hdir{64B}, \hdir{653}).

\exampleend
\end{example}

We start our analysis by defining the set $\mathcal{J}$, which
contains all possible jump destinations in an \EVM program
$P \equiv b_0,\dots, b_p$:
 \[ \mathcal{J}(P)=\{pc ~|~ b_{pc} \in P \wedge b_{pc} \equiv
\text{\code{JUMPDEST}}\}.
\]
We use $b_{pc}\in P$ for referring to the instruction at program
counter $pc$ in the \EVM program $P$.  In what follows, we omit $P$
from definitions when it is clear from the context, e.g., we use
$\mathcal{J}$ to refer to $\mathcal{J}(P)$.

\begin{example}
  Given the \EVM code that corresponds to function \code{findWinner},
  we get the following set:
\[
\mathcal{J} = \{
\hdir{123}, \hdir{142}, \hdir{954}, \hdir{64B}, \hdir{6D0}, \hdir{66F}, \hdir{653},
\hdir{6C3}, \hdir{691}, \hdir{6D1}, \hdir{6BA}
\}
\]
\exampleend
\end{example}

The first step in the computation of the CFG is to define the notion
of \emph{block}.
In general~\cite{AhoLSU07}, given a program $P$, a \emph{block} is a
maximal sequence of straight-line consecutive code in the program with
the properties that the flow of control can only enter the block
through the first instruction in the block, and can only leave the
block at the last instruction.  Let us define the concept of block in
an \EVM program:

\begin{definition}[blocks]\label{def:block}
Given an \EVM program $P \equiv b_0,\ldots,b_p$, we define

{\small 
\[ 
blocks(P) = \bigg\{B_i \equiv b_i,\ldots,b_j ~\bigg|~   
\begin{array}{l}
(\forall k. i<k<j, b_k \not\in Jump \cup End \cup
  \{\text{\code{JUMPDEST}}\}) ~\wedge\\
(~i {=} 0 \vee b_i {\equiv} \text{\code{JUMPDEST}} \vee  b_{i-1} {=}
  \text{\code{JUMPI}}~) ~\wedge\\ 
(~j {=} p \vee b_j \in Jump \vee b_j \in End \vee b_{j+1} {\equiv} \text{\code{JUMPDEST}}~)  
\end{array}
\bigg\}
\]
where 
\[
\begin{array}{rcl}
Jump & = & \{\text{\code{JUMP},\code{JUMPI}} \} \\
End & = &
\{
\text{\code{REVERT}},
\text{\code{STOP}},
\text{\code{INVALID}}
\} \\
\end{array}
\]
}
\end{definition}

\begin{figure}[t] %\vspace{-1cm}
 \begin{center}
\fbox{ \includegraphics[height=0.4\textheight,width=
0.95\textwidth]{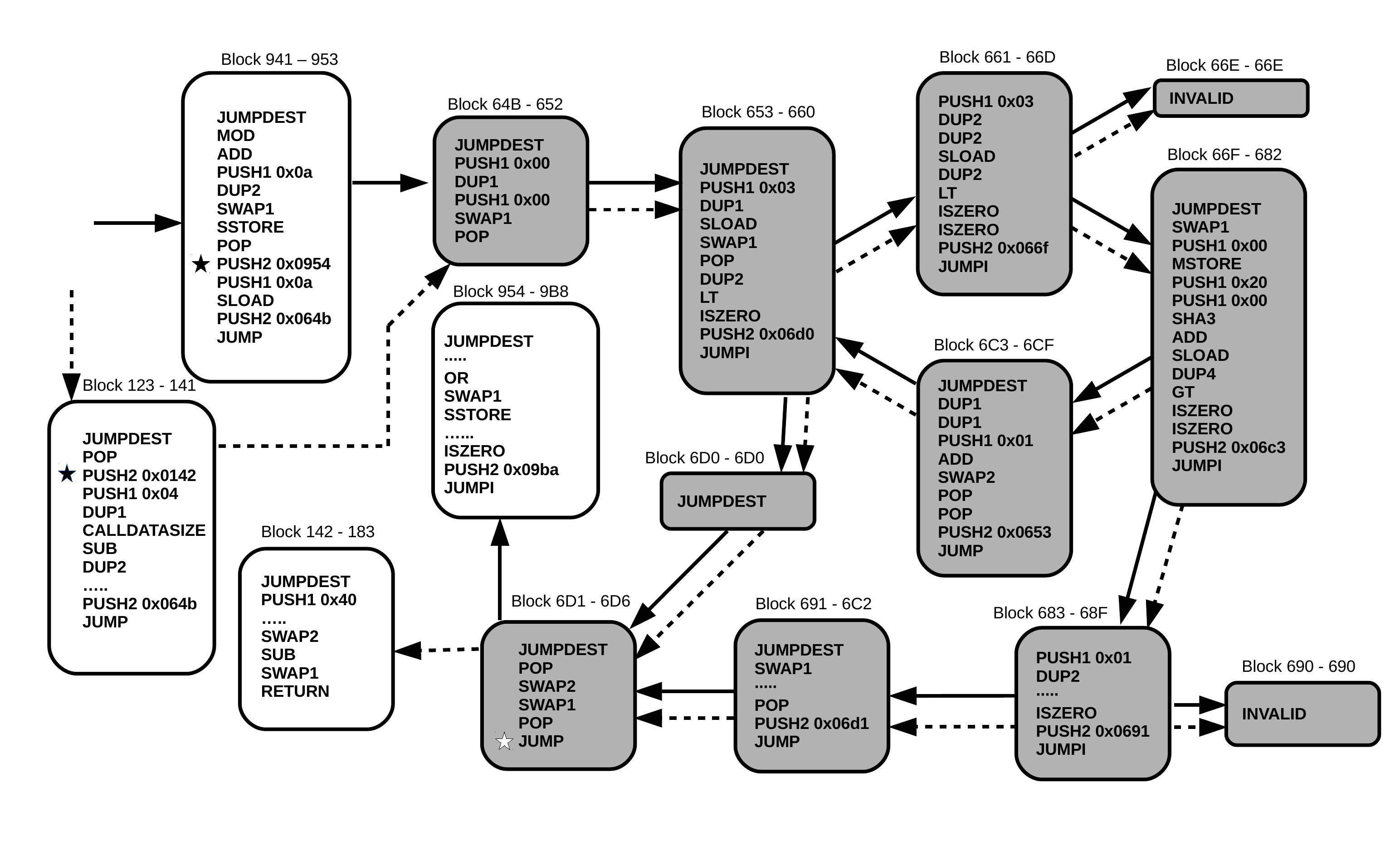}} 
\end{center}%\vspace{-1.3cm}  \end{center}
\caption{Fragment of the CFG of \code{findWinner}}
\label{fig:cfg-ins}
\end{figure}

\begin{example}
\label{ex:cfg-ins-blocks}

Figure~\ref{fig:cfg-ins} shows the blocks (nodes) obtained for
\code{findWinner} and their corresponding jump invocations.
% %
Solid and dashed edges represent the two possible execution paths
depending on the entry block: solid edges represent the path that
starts from block \hdir{941} and dashed edges the path that starts
from \hdir{123}.
% %
Note that most of the blocks start with a \code{JUMPDEST} instruction
(\hdir{123}, \hdir{941}, \hdir{64B}, \hdir{653}, \hdir{66F},
\hdir{954}, \hdir{6C3}, \hdir{691}, \hdir{142}, \hdir{6D1},
\hdir{6D0}).  The rest of the blocks start with instructions that come
right after a \code{JUMPI} instruction (\hdir{661}, \hdir{683}).
% %
Analogously, most blocks end in a \code{JUMP} (\hdir{941}, \hdir{6C3},
\hdir{123}, \hdir{691}, \hdir{6D1}), \code{JUMPI} (\hdir{653},
\hdir{661}, \hdir{66F}, \hdir{683}) or \code{RETURN} (\hdir{142})
instruction or in the instruction that precedes \code{JUMPDEST}
(\hdir{64B}).
% %
\exampleend
\end{example}

Observing the blocks in Figure~\ref{fig:cfg-ins}, we can see that most
\code{JUMP} instructions use the address introduced in the \code{PUSH}
instruction executed immediately before the \code{JUMP}. However, in
general, in \EVM code, it is possible to find a \code{JUMP} whose
address has been stored in a different block.
% %
This happens for instance when a public function is invoked privately
from other methods of the same contract, the returning program counter
is introduced by the invokers at different program points and it will
be used in a unique \code{JUMP} instruction when the invoked method
finishes in order to \emph{return} to the particular caller that
invoked that function.

\begin{example}\label{ex:why-clone}
  In Figure~\ref{fig:cfg-ins}, at block \hdir{6D1} we have a
  \code{JUMP} (marked with \ding{73}) whose address is not pushed in
  the same block. This \code{JUMP} takes the returned address from
  function \code{findWinner}. If \code{findWinner} is publicly
  invoked, it jumps to address \hdir{142} (pushed at block \hdir{123}
  at $\star$) and if it is invoked from \code{__callback} it jumps to
  \hdir{954} (pushed at block \hdir{941} at $\star$).
\end{example}

\subsection{Operational Semantics}

Figure~\ref{fig:jumpsemantics} shows the semantics of some
instructions involved in the computation of the values stored in the
stack for handling jumps.  The state of the program $S$ is a tuple
$\tuple{pc, \tuple{n,\sigma}}$ where $pc$ is the value of the program
counter with the index of the next instruction to be executed, and
$\tuple{n,\sigma}$ is a \emph{stack state} as defined in
Section~\ref{sec:from-evm-to-cfg} ($n$ is the number of elements in
the stack, and $\sigma$ is a partial mapping that relates some stack
positions with a set of jump destinations).
Interesting rules are the ones that deal with jump destination addresses
on the stack:  Rule \textsc{(4)} adds a new address on the stack, and
Rules  \textsc{(6)} and \textsc{(8-10)} copy or exchange existing
addresses on top of the stack, respectively. 
Rules \textsc{(1)} to \textsc{(3)} perform a jump in the program and
therefore consume the address placed on top of the stack, plus an
additional word in the case of \code{JUMPI}.  If the
instructions considered in this simplified semantics do not handle 
jump addresses, the corresponding rules just remove some values from
the stack in the program state $S$ (Rules \textsc{(5)}, \textsc{(7)}
and \textsc{(11)}).
The remaining \EVM instructions not explicitely considered in this
simplified semantics are generically represented by Rule \textsc{(12)}
with $b_{pc}^{\delta , \alpha}$, where $\delta$ is the number of items
removed from stack when $b_{pc}$ is executed, and $\alpha$ is the
number of additional items placed on the stack.
Complete executions are traces of the form
$S_0 \rrderivproof S_1 \rrderivproof \dots \rrderivproof S_n$ where $S_0 \equiv
\tuple{0, \tuple{0,\sigma_{\emptyset}}}$ is the initial
state, $\sigma_{\emptyset}$ is the empty mapping, and $S_n$ corresponds
to the last state.  There are no infinite traces, as any transaction
that executes \EVM code has a finite gas limit and every instruction
executed consumes some amount of gas.  When the gas limit is exceeded,
an out-of-gas exception occurs and the program halts immediately.

\input{semantics-al}

%% file: running-example.tex
\begin{figure}[h]
  \begin{minipage}{0.74\textwidth}
  \begin{center}
  \input pot-ex.tex

  \end{center}
  \end{minipage}
  \hspace{0.2cm}
  \begin{minipage}{0.23\textwidth}
   \input pot-ex-evm.tex
  \end{minipage} %\vspace{-0.8cm}
  \caption{Excerpt of \textsf{Solidity} code for \code{EthereumPot} contract (left), and fragment of \EVM code for
 function \code{findWinner} (right)}
\label{fig:solevm}
\end{figure}

%% file: pot-ex.tex
\begin{tabular}{l}
%\hline
\begin{lstlisting}[language=Solidity]
contract EthereumPot {
  address[] public addresses;
  address public winnerAddress;
  uint[] public slots;
  function __callback (bytes32 _queryId, string _result, bytes _proof){
    if (msg.sender != oraclize_cbAddress()) throw; 
    random_number = uint(sha3(_result))
    winnerAddress = findWinner(random_number);
    amountWon = this.balance * 98 / 100 ;
    winnerAnnounced(winnerAddress, amountWon);
    if (winnerAddress.send(amountWon)) {
       if (owner.send(this.balance)) {
         openPot();
       }
    }
  }

  function findWinner (uint random) constant returns (address winner) { 
    for (uint i = 0; i < slots.length; i++) {
      if (random <= slots[i]) {
        return addresses[i];
      }
    }
  }
  // Other functions
}
\end{lstlisting}
% \\ 
% \hline 
\end{tabular}
 % uint public random_number;
 % $\cdots$
 % uint public amountWon;
 % $\cdots$  

%  function /* actual callback */__callback(bytes32 _queryId, string _result, bytes _proof) oraclize_randomDS_proofVerify(_queryId, _result, _proof)
%  {
%    // if we reach this point successfully, it means that the attached authenticity proof has passed!
%    if(msg.sender != oraclize_cbAddress()) throw; 
%    
%    // generate a random number between potSize(number of tickets sold) and 1
%    random_number = uint(sha3(_result))%potSize + 1;
%
%    // find that winner based on the random number generated
%    winnerAddress = findWinner(random_number);
%
%    // winner wins 98% of the remaining balance after oraclize fees
%    amountWon = this.balance * 98 / 100 ;
%
%    // announce winner
%    winnerAnnounced(winnerAddress, amountWon);
%    if(winnerAddress.send(amountWon)) {
%
%      if(owner.send(this.balance)) {
%        openPot();
%      }
%    }
%  }

%% file: pot-ex-evm.tex
\begin{tabular}{l}
%block64B %block653 %block907 %block941
\hspace{3pt}
{\begin{lstlisting}[basicstyle=\scriptsize\ttfamily,numbers=none]
$\cdots$
64B: JUMPDEST
64C: PUSH1 0x00
64E: DUP1
64F: PUSH1 0x00
651: SWAP1
652: POP
653: JUMPDEST
654: PUSH1 0x03
656: DUP1
657: SLOAD
658: SWAP1
659: POP
65A: DUP2
65B: LT
65C: ISZERO
65D: PUSH2 0x06D0
660: JUMPI
661: PUSH1 0x03
663: DUP2
664: DUP2
665: SLOAD
666: DUP2
$\cdots$
941: JUMPDEST
942: MOD
943: ADD
944: PUSH1 0x0A
946: DUP2
947: SWAP1
948: SSTORE
949: POP
94A: PUSH2 0x0954
94D: PUSH1 0x0A
94F: SLOAD
950: PUSH2 0x064B
953: JUMP
$\cdots$
\end{lstlisting}}
\end{tabular}

%% file: semantics-al.tex
% ;; -*- coding: iso-latin-1; TeX-PDF-mode: t; TeX-master: "main" -*-%

\begin{figure}[h]%\centering
{\small
  \begin{tabular}{|p{0.3cm}p{14cm}|}
\hline
\textsc{(1)} & 
      \rrrulewide{b_{pc} = \text{\code{JUMP}}}
      {\tuple{pc, \tuple{n,\sigma}}
        \rrderivproof
        \tuple{\sigma(s_{n-1}) , \tuple{n-1, \sigma\backslash[s_{n-1}]}}}
\ruleskip
\textsc{(2)} & 
      \rrrulewide{b_{pc} = \text{\code{JUMPI}}}
      {\tuple{pc, \tuple{n,\sigma}}
        \rrderivproof
        \tuple{\sigma(s_{n{-}1}), \tuple{n{-}2, \sigma\backslash[s_{n{-}1},s_{n{-}2}]}}}
\ruleskip
\textsc{(3)} & 
      \rrrulewide{b_{pc} = \text{\code{JUMPI}}}
      {\tuple{pc, \tuple{n,\sigma}}
        \rrderivproof
        \tuple{pc{+}size(b_{pc}), \tuple{n{-}2, \sigma\backslash[s_{n{-}1},s_{n{-}2}]}}}
\ruleskip

\textsc{(4)} & 
      \rrrulewide{b_{pc} = \text{\code{PUSH}}x~v, \text{v} \in \mathcal{J}} 
      {\tuple{pc, \tuple{n,\sigma}}
        \rrderivproof
        \tuple{pc{+}size(b_{pc}), \tuple{n+1, \sigma[s_n \mapsto \{v\}]}}}
\ruleskip

\textsc{(5)} & 
      \rrrulewide{b_{pc} = \text{\code{PUSH}}x~v, v \notin \mathcal{J}}
      {\tuple{pc, \tuple{n,\sigma}}
        \rrderivproof
        \tuple{pc{+}size(b_{pc}), \tuple{n+1, \sigma}}}
\ruleskip

\textsc{(6)} & 
      \rrrulewide{b_{pc} = \text{\code{DUP}}x, s_{n {-} x} \in dom(\sigma)}
      {\tuple{pc, \tuple{n,\sigma}}
        \rrderivproof
        \tuple{pc{+}size(b_{pc}), \tuple{n+1, \sigma[s_n \mapsto \sigma(s_{n {-} x})]}}}
\ruleskip

\textsc{(7)} & 
      \rrrulewide{b_{pc} = \text{\code{DUP}}x, s_{n {-} x} \notin dom(\sigma)}
      {\tuple{pc, \tuple{n,\sigma}}
        \rrderivproof
        \tuple{pc{+}size(b_{pc}), \tuple{n+1, \sigma}}}
      \ruleskip

\textsc{(8)} & 
      \rrrulewide{b_{pc} = \text{\code{SWAP}}x,  s_{n {-} 1} \in dom(\sigma),  s_{n {-} x {-} 1} \in dom(\sigma)}
      {\tuple{pc, \tuple{n,\sigma}}
        \rrderivproof 
        \tuple{pc{+}size(b_{pc}), \tuple{n, \sigma[ s_{n {-} x {-} 1} \mapsto \sigma(s_{n {-} 1}),  s_{n {-} 1} \mapsto  \sigma(s_{n {-} x {-} 1})]}}}~~~~~~
      \ruleskip

\textsc{(9)} & 
      \rrrulewide{b_{pc} = \text{\code{SWAP}}x,  s_{n {-} 1} \in dom(\sigma),  s_{n {-} x {-} 1} \notin dom(\sigma)}
      {\tuple{pc, \tuple{n,\sigma}}
        \rrderivproof
        \tuple{pc{+}size(b_{pc}), \tuple{n, \sigma[s_{n {-} x {-} 1} \mapsto \sigma(s_{n {-} 1})]\backslash[s_{n {-} 1}]}}}
\ruleskip

\textsc{(10)} & 
      \rrrulewide{b_{pc} = \text{\code{SWAP}}x,  s_{n {-} 1} \notin dom(\sigma),  s_{n {-} x {-} 1} \in dom(\sigma)}
      {\tuple{pc, \tuple{n,\sigma}}
        \rrderivproof
        \tuple{pc{+}size(b_{pc}), \tuple{n, \sigma[s_{n {-} 1} \mapsto \sigma(s_{n {-} x {-} 1})]\backslash[s_{n {-} x {-} 1}]}}}
      \ruleskip

\textsc{(11)} & 
      \rrrulewide{b_{pc} = \text{\code{SWAP}}x,  s_{n {-} 1} \notin dom(\sigma),  s_{n {-} x {-} 1} \notin dom(\sigma)}
      {\tuple{pc, \tuple{n,\sigma}}
        \rrderivproof
        \tuple{pc{+}size(b_{pc}), \tuple{n, \sigma\backslash[s_{n{-}1}, s_{n {-} x {-} 1}] }}}
      \ruleskip

\textsc{(12)} & 
      \rrrulewide{%b_{pc}^{\delta , \alpha} \in \textit{otherwise},
                b_{pc}^{\delta , \alpha} \notin \textit{End} \cup
                \textit{Jump} \cup \{\text{\code{PUSH}}x,
                \text{\code{DUP}}x, \text{\code{SWAP}}x \}}
      {\tuple{pc, \tuple{n,\sigma}}
        \rrderivproof
        \tuple{pc{+}size(b_{pc}), \tuple{n {-} \delta {+} \alpha, \sigma\backslash[s_{n-1}, ..., s_{n - \delta}]}}}
\\
      \hline
  \end{tabular}
}
\caption{Simplified \EVM semantics for handling jumps\label{EVM_semantics}}
\label{fig:jumpsemantics}
\end{figure}

%% file: cfganalysis.tex
% ;; -*- coding: iso-latin-1; TeX-PDF-mode: t; TeX-master: "main" -*-%

\section{From \EVM to a Sound CFG}
\label{sec:from-evm-to-cfg}

As we have seen in the previous section, the addresses used by the
jumping instructions are stored in the execution stack. In \evm, blocks
can be reached with different stack sizes an contents. As it is used
in other tools~\cite{gigahorse,vandal,madmax}, to precisely infer the
possible addresses at jumping program points, we need a
\emph{context-sensitive} static analysis that analyze separately all
blocks for each possible stack than can reach them (only considering
the addresses stored in the stack).
This section presents an \emph{address analysis} of \EVM programs
%  (see
% Figure~\ref{fig:tool}) 
which allows us to compute a complete CFG of the
\EVM code.
% %
To compute the addresses involved in the jumping instructions, we
define a static analysis which soundly infers all possible addresses
that a \code{JUMP} instruction could use.
% %

% %
In our address analysis we aim at having the stack represented by
explicit variables. Given the characteristics of \EVM programs, the
execution stack of \EVM programs produced from \solidity programs
without recursion can be flattened. Besides, as the size of the stack
of the Ethereum Virtual Machine is bounded to 1024 elements
(see~\cite{yellow}), the number of stack variables is limited. We use
$\mathcal{V}$ to represent the set of all possible stack variables
that may be used in the program.
% %
The first element we define for our analysis is its abstract state:

\paragraph{The abstract state}
Our analysis uses a partial representation of the execution stack as
basic element.  To this end, we use the notion of \emph{stack state}
as a pair $\tuple{n,\sigma}$, where $n$ is the number of elements in
the stack, and $\sigma$ is a partial mapping that relates some stack
positions with a set of jump destinations.  A position in the stack is
referred as $s_i$ with $0 \leq i < n$, and $s_{n-1}$ is the position
at the top of the stack.
The \emph{abstract state} of the analysis is defined on the set of all
stack states
$\mathcal{S} = \{ \langle n, \sigma \rangle ~|~ 0 \leq n \leq
|\mathcal{V}| \wedge \sigma(s) \in \Sigma_n \}$ where $\Sigma_n$ is
the set of all mappings using up to $n$ stack variables.
% , defined recursively
% as follows:
% $\Sigma_i = \Sigma_{i-1} \cup \{\sigma[s_i \mapsto j] ~|~ \sigma
% \in\Sigma_{i-1} \wedge j \subseteq \mathcal{J}\}$;
% $\Sigma_0 = \{\sigma_{\emptyset}\}$, where $\sigma_{\emptyset}$ is the
% empty mapping.

\begin{definition}[abstract state]
% %
  The abstract state is a partial \emph{mapping} $\pi$ of the form
  $\mathcal{S} \mapsto \mathcal{P}(\mathcal{S})$. 
%   , where
%   $n\in \mathbb{N}, n \leq |\mathcal{V}|,$ and $\sigma$ is a partial
%   mapping $\sigma : \mathcal{V} \mapsto \mathcal{J}$.
\end{definition}

The application of $\sigma$ to an element $s_i$, that is, $\sigma(s_i)$,
corresponds to the set of jump destinations that a stack variable $s_i$ can
contain. The first element of the tuple, that is, $n$, stores the size of
the stack in the different abstract states.

The abstract domain is the lattice $\langle \it AS, \ptabstracttop,
\ptabstractbot, \sqcup, \sqsubseteq\rangle$, where $\it AS$ is the set of
abstract states and $\ptabstracttop$ is the top of the lattice defined as
% %
the mapping $\pi_\top$ such that $\forall s \in \mathcal{S}, \pi_{\top}(s) =
\mathcal{S}$.
% %$\forall s_1,s_2 \in \mathcal{V}, \pi(s_1) = \langle \sigma(s_2) =
% %\mathcal{J} \rangle$.
% %.
The bottom element of the lattice $\ptabstractbot$ is the empty mapping.
% %
Now, to define $\sqcup$ and $\sqsubseteq$, we first define the function
$img(\ptabstract,s)$ as $\ptabstract(s)$ if $s \in dom(\pi)$ and $\emptyset$,
otherwise. 
Given two abstract states $\ptabstract_1$ and $\ptabstract_2$, we
use $\ptabstract = \ptabstract_1 \lub \ptabstract_2$ to denote that $\ptabstract
$ is the least upper-bound defined as follows
$
\forall s \in dom(\ptabstract_1) \cup dom(\ptabstract_2), 
\ptabstract(s) = img(\ptabstract_1,s) \cup img(\ptabstract_2,s)$.
%
%%
% \[\forall s \in dom(\pi_1) \cup dom(\pi_2) ~|~ \pi(s) = \langle
%   \sigma,n\rangle,
% \]
% \[ \begin{cases} \pi(s) = \pi_1(s) & \text{if } s \in dom(\pi_1)
%     \wedge s \not\in dom(\pi_2)
%     \\
%     \pi(s) = \pi_2(s) & \text{if } s \not\in dom(\pi_1) \wedge s \in
%     dom(\pi_2)
%     \\
%     \pi(s) = \langle \sigma_1 \uplus \sigma_2, n\rangle, & \text{if
%   } s \in dom(\pi_1) \wedge s \in dom(\pi_2)
% \end{cases}
% \]
% where $\sigma = \sigma_1 \uplus \sigma_2$ is defined as 
% \[
% \forall s  \in dom(\sigma_1) \cup dom(\sigma_2),\]
% \[
% \begin{cases}
%   \sigma(s) = \sigma_1(s) & \text{if } s \in dom(\sigma_1) \wedge s
%   \not\in dom(\sigma_2)
%   \\
%   \sigma(s) = \sigma_2(s) & \text{if } s \not\in dom(\sigma_1)
%   \wedge s \in dom(\sigma_2)
%   \\
%   \sigma(s) = \sigma_1(s) \cup \sigma_2(s), & \text{if } s \in
%   dom(\sigma_1) \wedge s \in dom(\sigma_2)
% \end{cases}
% \]
% %
At this point, $\ptabstract_1 \sqsubseteq \ptabstract_2$ holds iff
$dom(\ptabstract_1) \subseteq dom(\ptabstract_2)$ and $\forall s \in
dom(\pi_1), \pi_1(s) \subseteq \pi_2(s). $
%% 
% \[
% \begin{array}{c}
%  % dom(\pi_1) \subseteq dom(\pi_2) ~\wedge \\[0.2cm]
%  % \forall s \equiv \langle \sigma_1,n
%  % \rangle \in dom(\pi_1)  ~|~ \wedge \pi_2(s) \equiv \langle \sigma_2,n \rangle
%  % \wedge dom(\sigma_1) \subseteq dom(\sigma_2) ~\wedge \\[0.2cm] 
%  % (\forall v \in dom(\sigma_1) 
%  ~|~ \sigma_1(v) \subseteq \sigma_2(v))
% \end{array}
% \]

\paragraph{Transfer function} One of the ingredients of our analysis
is a \emph{transfer function} that models the effect of each \EVM
instruction on the abtract state for the different instructions.
% As the abstract state is a set of elements of the form
% $s = \langle n,\sigma\rangle$, we first define how these elements
% are updated according to the \EVM instruction we are processing.
Given a stack state $s$ of the form $\langle n, \sigma \rangle$,
Figure~\ref{fig:transfer} defines the updating function
$\lambda(b_{pc}^{\delta,\alpha},s)$ where $b$ corresponds to the \EVM
instruction to be applied, $pc$ corresponds to the program counter of
the instruction and $\alpha$ and $\delta$ to the number of elements
placed to and removed from the \EVM stack when executing $b$,
respectively.
Given a map $m$ we will use $m[x\mapsto y]$ to indicate the result of
updating $m$ by making $m(x) = y$ while $m$ stays the same for all
locations different from $x$, and we will use $m\backslash[x]$ to
refer to a partial mapping that stays the same for all locations
different from $x$, and $m(x)$ is undefined.
%
% \begin{definition}[updating function]
%   Given a tuple of the form $\langle \sigma, \mathbb{N}\rangle$ and
%   an \EVM instruction $Ins$, the updating function $\lambda$ is
%   defined as a mapping of the form:
% \[
%   \lambda~:~ Ins \times \langle \sigma, \mathbb{N}\rangle \mapsto
%   \langle \sigma', \mathbb{N}\rangle
% \]
% computed according to the table in Figure~\ref{fig:transfer}.
% \end{definition}
By means of $\lambda$, we define the \emph{transfer function}
of our analysis.

\begin{definition}[transfer function]
\label{def:transfer}
Given the set of abstract states $AS$ and the set of \EVM instructions
$Ins$, the transfer function $\tau$ is defined as a mapping of the
form
\[
\tau~:~ Ins \times AS \mapsto AS
\]
is defined as follows:
\[
\tau (b,\pi) = \pi' ~~~~~~ \text{where} ~~~ 
\forall s \in dom(\pi), \pi'(s) = \lambda(b,\pi(s))
\]
\end{definition}

\input transfer

\begin{example}

  Given the following initial abstract state
  $\fstate{8,\set{}}{8,\set{}}$, which corresponds to the initial
  stack state for executing block \hdir{941}, the application of the
  transfer function $\tau$ to the block that starts at \EVM
  instruction \hdir{941}, produces the following results (between
  parenthesis we show the program point). To the right we show the
  application of the transfer function to block \hdir{123} with its
  initial abstract state $\fstate{2,\set{}}{2,\set{}}$.

\noindent
\begin{minipage}{7.7cm} 
{\scriptsize
\[
\begin{array}{lll}
(\hdir{941}) & \text{\code{JUMPDEST}} & \fstate{8,\set{}}{8,\set{}}\\
(\hdir{942}) & \text{\code{MOD}} & 			\fstate{8,\set{}}{7,\set{}} \\
(\hdir{943}) & \text{\code{ADD}} & 			\fstate{8,\set{}}{6,\set{}} \\
(\hdir{944}) & \text{\code{PUSH1 0A}} & 	\fstate{8,\set{}}{7,\set{}} \\
(\hdir{946}) & \text{\code{DUP2}} & 			\fstate{8,\set{}}{8,\set{}} \\
(\hdir{947}) & \text{\code{SWAP1}} & 		\fstate{8,\set{}}{8,\set{}} \\
(\hdir{948}) & \text{\code{SSTORE}} & 	\fstate{8,\set{}}{6,\set{}} \\
(\hdir{949}) & \text{\code{POP}} & 			\fstate{8,\set{}}{5,\set{}} \\
(\hdir{94A}) & \text{\code{PUSH2 0954}} & 	\fstate{8,\set{}}{6,\set{s_5 \mapsto
\hdir{954}}\}} \\
(\hdir{94D}) & \text{\code{PUSH1 0A}} & 	\fstate{8,\set{}}{7,\set{s_6 \mapsto \hdir{954}}\}} \\
(\hdir{94F}) & \text{\code{SLOAD}} & 		\fstate{8,\set{}}{7,\set{s_5 \mapsto \hdir{954}}\}} \\
(\hdir{950}) & \text{\code{PUSH2 064B}} & 	\fstate{8,\set{}}{8,\set{s_5 \mapsto
\hdir{954},s_7\mapsto \hdir{64B}}}\\
(\hdir{953}) & \text{\code{JUMP}} & 			\fstate{8,\set{}}{7,\set{s_5 \mapsto \hdir{954}\}}}
\\
\end{array}
\]	
}
\end{minipage}
\begin{minipage}{8cm} 
{\scriptsize
\[
\begin{array}{lll}
(\hdir{123}) & \text{\code{JUMPDEST}} & 		\fstate{2,\set{}}{2,\set{}} \\
(\hdir{124}) & \text{\code{POP}} & 			\fstate{2,\set{}}{1,\set{}} \\
(\hdir{125}) & \text{\code{PUSH2 0142}} & 	\fstate{2,\set{}}{2,\set{s_1 \mapsto
\hdir{142}}} \\
(\hdir{128}) & \text{\code{PUSH1 04}} & 	\fstate{2,\set{}}{3,\set{s_1 \mapsto \hdir{142}}} \\
(\hdir{12A}) & \text{\code{DUP1}} & 			\fstate{2,\set{}}{4,\set{s_1 \mapsto \hdir{142}}} \\
(\hdir{12B}) & \text{\code{CALLDATASIZE}} & 	\fstate{2,\set{}}{5,\set{s_1 \mapsto \hdir{142}}}\\
(\hdir{12C}) & \text{\code{SUB}} & 			\fstate{2,\set{}}{4,\set{s_1 \mapsto \hdir{142}}} \\
~~~~~~~\vdots \\
% (\hdir{12D}) & \text{\code{DUP2}} & 			\fstate{2,\set{}}{5,\set{s_1 \mapsto \hdir{142}}} \\
% (\hdir{12E}) & \text{\code{ADD}} & 			\fstate{2,\set{}}{4,\set{s_1 \mapsto \hdir{142}}} \\
% (\hdir{12F}) & \text{\code{SWAP1}} & 		\fstate{2,\set{}}{4,\set{s_1 \mapsto \hdir{142}}} \\
% (\hdir{130}) & \text{\code{DUP1}} & 			\fstate{2,\set{}}{5,\set{s_1 \mapsto \hdir{142}}} \\
% (\hdir{131}) & \text{\code{DUP1}} & 			\fstate{2,\set{}}{6,\set{s_1 \mapsto \hdir{142}}} \\
% (\hdir{132}) & \text{\code{CALLDATALOAD}} & 	\fstate{2,\set{}}b{6,\set{s_1 \mapsto \hdir{142}}}\\
% (\hdir{133}) & \text{\code{SWAP1}} & 		\fstate{2,\set{}}{6,\set{s_1 \mapsto \hdir{142}}} \\
% (\hdir{134}) & \text{\code{PUSH 0x20}} & 		\fstate{2,\set{}}{7,\set{s_1 \mapsto \hdir{142}}} \\
% (\hdir{136}) & \text{\code{ADD}} & 			\fstate{2,\set{}}{6,\set{s_1 \mapsto \hdir{142}}} \\
% (\hdir{137}) & \text{\code{SWAP1}} & 		\fstate{2,\set{}}{5,\set{s_1 \mapsto \hdir{142}}} \\
% (\hdir{138}) & \text{\code{SWAP3}} & 		\fstate{2,\set{}}{5,\set{s_1 \mapsto \hdir{142}}} \\
% (\hdir{139}) & \text{\code{SWAP2}} & 		\fstate{2,\set{}}{5,\set{s_1 \mapsto \hdir{142}}} \\
(\hdir{13A}) & \text{\code{SWAP1}} & 		\fstate{2,\set{}}{5,\set{s_1 \mapsto \hdir{142}}} \\
(\hdir{13B}) & \text{\code{POP}} & 			\fstate{2,\set{}}{4,\set{s_1 \mapsto \hdir{142}}} \\
(\hdir{13C}) & \text{\code{POP}} & 			\fstate{2,\set{}}{3,\set{s_1 \mapsto \hdir{142}}} \\
(\hdir{13D}) & \text{\code{POP}} & 			\fstate{2,\set{}}{2,\set{s_1 \mapsto \hdir{142}}} \\
(\hdir{13E}) & \text{\code{PUSH2 064B}} &
\fstate{2,\set{}}{3,\set{s_1\mapsto\hdir{142},s_2\mapsto \hdir{64B}}}\\
(\hdir{141}) & \text{\code{JUMP}} & 			\fstate{2,\set{}}{2,\set{s_1 \mapsto \hdir{142}}} \\
\end{array}
\]	
}
\end{minipage}

\exampleend
\end{example}

\subsection{Addresses equation system}
\label{eq:equation_system}

The next step consists in defining, by means of the transfer and the updating
functions, a constraint equation system to represent all possible jumping
addresses that could be valid for executing a \emph{jump} instruction in the
program.

\begin{definition}[addresses equation system]
\label{def:jumpeq}
Given an \evm program $P$ of the form $b_0,\ldots,b_p$, its \emph{addresses 
equation system}, $\mathcal{E}(P)$ includes the following equations
according to all \EVM bytecode instruction $b_{pc} \in P$:

\bigskip
\begin{center}
{\small
\begin{tabular}{rp{3.5cm}|p{1.7cm}p{0.3cm}lr}
\cline{2-6}
& 
\multicolumn{1}{c}{$b_{pc}$} &
\multicolumn{4}{c}{$C_{pc}$}
 \\
\cline{2-6}

(1) &
\code{JUMP} &
$\eqn_{\sigma(s_{n-1})}$ &
$\sqsupseteq$ &
%$\tau(\text{\code{JUMP}},\eqn{i})$ &~~~~~~~ $

$\text{\textit{idmap}}(\lambda(b_{pc}, \langle n, \sigma \rangle))$
~~~~~~~~~~~~ & 
$\forall s \in dom(\eqn_{pc}), \langle n,\sigma\rangle \in \eqn_{pc}(s)$
% $\{\langle n{-}1, \sigma\backslash[s_{n{-}1}]\rangle  ~|~ \langle n,\sigma\rangle \in \eqi\}$ \\[0.2cm]
%%%$\{\langle n', \sigma'\rangle  ~|~ \langle
%%%n,\sigma\rangle \in \eqi \wedge \langle n', \sigma'\rangle =
%%%\lambda(\text{\code{JUMP}},\langle n, \sigma\rangle)\}$
\\[0.2cm]
\cline{2-6}

\multirow{2}{*}{(2)} &
\multirow{2}{*}{\code{JUMPI}} &
$\eqn_{\sigma(s_{n-1})}$ &
$\sqsupseteq$ &
$\text{\textit{idmap}}(\lambda(b_{pc}, \langle n, \sigma \rangle))$
& 
$\forall s \in dom(\eqn_{pc}), \langle n,\sigma\rangle \in \eqn_{pc}(s)$
%$\{\langle n{-}2, \sigma\backslash[s_{n{-}1},s_{n{-}2}]\rangle] ~|~ \langle n,
%\sigma\rangle \in \eqn_{i}\}$ \\[0.2cm]
%%%$\{\langle n', \sigma'\rangle  ~|~ \langle
%%%n,\sigma\rangle \in \eqn_i \wedge \langle n', \sigma'\rangle =
%%%\lambda(\text{\code{JUMP}},\langle n, \sigma\rangle)\}$
\\[0.1cm]
 &
 &
$\eqn_{pc+1}$ &
$\sqsupseteq$ &
$\text{\textit{idmap}}(\lambda(b_{pc}, \langle n, \sigma \rangle))$
%$\tau(b_i,\eqn_{i})$ 
& 
$\forall s \in dom(\eqn_{pc}), \langle n,\sigma\rangle \in \eqn_{pc}(s)$

\\[0.2cm]
\cline{2-6}

(3) &
$b_{pc} \not\in End \wedge$ \newline 
$b_{pc+size(b_{pc})} = \text{\code{JUMPDEST}}$ & 
$\eqn_{pc+size(b_{pc})}$ &
$\sqsupseteq$ &
%$\tau(\text{\code{JUMP}},\eqn{i})$ &~~~~~~~ $

$\text{\textit{idmap}}(\lambda(b_{pc}, \langle n, \sigma \rangle))$
~~~~~~~~~~~~ & 
$\forall s \in dom(\eqn_{pc}), \langle n,\sigma\rangle \in \eqn_{pc}(s)$

\\[0.2cm]
\cline{2-6}
 
(4) &
$b_{pc} \not\in End$ &
$\eqn_{pc+size(b_{pc})}$ &
$\sqsupseteq$ &
$\tau(b_{pc},\eqn_{pc})$ \\[0.2cm]
\cline{2-6}

(5) &
otherwise&
$\eqn_{pc+size(b_{pc})}$ &
$\sqsupseteq$ &
$\tau(b_{pc},\eqn_{pc})$ \\[0.2cm]
\cline{2-6}
 
\end{tabular}
}
\end{center}

\medskip
where $idmap(s)$ returns a map $\pi$ such that $dom(\pi) = \{s\}$ and
$\pi(s) = \{s\}$ and $size(b_{pc})$ returns the number of bytes of the
instruction $b_{pc}$.

\end{definition}

Observe that the addresses equation system will have equations for
all program points of the program. Concretely, variables of the form
$\eqn_{pc}$ store the \emph{jumping addresses} saved in the stack
after executing $b_{pc}$ for all possible entry stacks.  This
information will be used for computing all possible jump destinations
when executing \code{JUMP} or \code{JUMPI} instructions.
% %
For computing the system, most instructions, cases (4) and (5), just
apply the transfer function $\tau$ to compute the possible stack
states of the subsequent instruction. Note that the expression
$pc+size(b_{pc})$ at (3) just computes the position of the next
instruction in the \EVM program.
% %
Jumping instructions, points (1) and (2), compute the initial state of
the invoked blocks, thus they produce a map with all possible input
stack states that can reach one block.  \code{JUMP} and \code{JUMPI}
instructions produce, for each stack state, one equation by taking the
element from the previous stack state $\eqn_{\sigma(s_{n-1})}$.
\code{JUMPI}, point (2), produces an extra equation $\eqn_{pc+1}$ to
capture the possibility of continuing to the next instruction instead
of jumping to the destination address. Additionally, those
instructions before \code{JUMPDEST}, point (3), produce initial states
for the block that starts in the \code{JUMPDEST}.
% %
When the constraint equation system is solved, constraint variables
over-approximate the jumping information for the program.

\begin{example}
\label{ex:transfer}

As it can be seen in Figure~\ref{fig:cfg-ins}, we can jump to block 64B from two
different blocks, \hdir{941} and \hdir{123}. The computation of the jump
equations systems will produce the following equations for the entry program
points of these two blocks:
% %

{
\small
\noindent
\begin{minipage}{8cm}
\[
\begin{array}{rcl}
\eq{941} & 
\sqsupseteq & 
\fstate{8,\set{}}{8,\set{}}
\\ 
& \vdots 
\\
\eq{950} & 
\sqsupseteq & 
\fstate{8,\set{}}{8,\{s_5 \mapsto \hdir{954}, s_8 \mapsto \hdir{64B}\}}
\\ 

\eq{64B}^{\anno{1}} & 
\sqsupseteq & 
\fstate{7,\set{s_5 \mapsto \hdir{954}}}{7,\set{s_5 \mapsto \hdir{954}}}
\end{array}
\]
\end{minipage}
\begin{minipage}{7cm}
\[
\begin{array}{rcl}
\eq{123} & 
\sqsupseteq & 
\fstate{2,\set{}}{2,\{\}}
\\ 
& \vdots
\\
\eq{318} & 
\sqsupseteq & 
\fstate{2,\set{}}{4,\{s_1 \mapsto \hdir{142}, s_3 \mapsto \hdir{64B}\}}
\\

\eq{64B}^{\anno{2}} & 
\sqsupseteq & 
\fstate{3,\set{s_1 \mapsto \hdir{142}}}{3,\set{s_1 \mapsto \hdir{142}}}

\end{array}
\]
\end{minipage}
}%%

\bigskip
\noindent Observe that we have two different stack contents
reaching the same program point, e.g. two equations for $\eq{64B}$ are
produced by two different blocks, the \code{JUMP} at the end of block
\hdir{941}, identified by $\eq{64B}^\anno{1}$, and the \code{JUMP} at
the end of block 123, identified by $\eq{64B}^\anno{2}$.
% %
Thus the equation that must hold for
% , that the right-hand side of equation
$\eq{64B}$ is produced by the application of the operation $\eq{64B}^\anno{1}
\lub \eq{64B}^\anno{2}$,
% which returns the following abstract state:
as follows:

{\small
\[
\eq{64B} \sqsupseteq 
\{
\tuple{7,\set{s_5 \mapsto \hdir{954}}} \mapsto \tuple{7,\set{s_5 \mapsto \hdir{954}}},
\tuple{3,\set{s_1 \mapsto \hdir{142}}} \mapsto \tuple{3,\set{s_1 \mapsto \hdir{142}}}
\}
\]
}
Note that the application of the transfer function $\tau$ for all instructions
of block 64B applies function $\lambda$ to all elements in the abstract state
and updates the stack state accordingly

\begin{center}
{\small
\[
\begin{array}{rrcll}  
(\text{\code{JUMPDEST}}) & \eq{64B} & \sqsupseteq &
\{ 		
\tuple{7,\set{s_5 \mapsto \hdir{954}}} \mapsto \tuple{7,\set{s_5 \mapsto \hdir{954}}}, &
\tuple{3,\set{s_1 \mapsto \hdir{142}}} \mapsto \tuple{3,\set{s_1 \mapsto \hdir{142}}} 
\}
\\

(\text{\code{PUSH1 00}}) & \eq{64C} &	 	\sqsupseteq &
\{ 		
\tuple{7,\set{s_5 \mapsto \hdir{954}}} \mapsto \tuple{8,\set{s_5 \mapsto \hdir{954}}},&
\tuple{3,\set{s_1 \mapsto \hdir{142}}} \mapsto \tuple{4,\set{s_1 \mapsto \hdir{142}}}
\}
\\

(\text{\code{DUP1}}) & \eq{64E} & \sqsupseteq & 			
\{ 		
\tuple{7,\set{s_5 \mapsto \hdir{954}}} \mapsto \tuple{9,\set{s_5 \mapsto \hdir{954}}}, &
\tuple{3,\set{s_1 \mapsto \hdir{142}}} \mapsto \tuple{5,\set{s_1 \mapsto \hdir{142}}}   
\}
\\

(\text{\code{PUSH1 00}}) & \eq{64F} & \sqsupseteq & 	
\{ 		
\tuple{7,\set{s_5 \mapsto \hdir{954}}} \mapsto \tuple{10,\set{s_5 \mapsto \hdir{954}}}, &
\tuple{3,\set{s_1 \mapsto \hdir{142}}} \mapsto \tuple{6,\set{s_1 \mapsto \hdir{142}}}
\}
\\

(\text{\code{SWAP1}}) & \eq{651} & \sqsupseteq &		
\{ 		
\tuple{7,\set{s_5 \mapsto \hdir{954}}} \mapsto \tuple{10,\set{s_5 \mapsto \hdir{954}}}, &
\tuple{3,\set{s_1 \mapsto \hdir{142}}} \mapsto \tuple{6,\set{s_1 \mapsto \hdir{142}}}
\}
\\

(\text{\code{POP}}) & \eq{652} & \sqsupseteq &			
\{ 		
\tuple{7,\set{s_5 \mapsto \hdir{954}}} \mapsto \tuple{9,\set{s_5 \mapsto \hdir{954}}}, &
\tuple{3,\set{s_1 \mapsto \hdir{142}}} \mapsto \tuple{5,\set{s_1 \mapsto \hdir{142}}}
\}
\\
\end{array}
\]	
}
\end{center}
\exampleend
\end{example}

% %
Solving the addresses equation system of a program $P$ can be done iteratively.
A na\"{\i}ve algorithm consists in first creating one constraint variable
$\eqn_0 \sqsupseteq \pi_{\emptyset}[\tuple{0,\sigma_{\emptyset}} \mapsto
\set{\tuple{0,\sigma_{\emptyset}}}]$, where $\pi_{\emptyset}$ and
$\sigma_{\emptyset}$ are empty mappings, and
% %
$\eqn_{pc} \sqsupseteq \ptabstractbot$ for all $pc \in P, pc \not=0$, and then
iteratively refining the values of these variables as follows:
\begin{enumerate}
\item substitute the current values of the constraint variables in
  the right-hand side of each constraint, and then evaluate the right-hand 
  side if needed;

\item if each constraint $\eqn \sqsupseteq E$ holds, where $E$ is the value
  of the evaluation of the right-hand side of the previous step, then
  the process finishes; otherwise

\item for each $\eqn \sqsupseteq E$ which does not hold, let $E'$ be the
  current value of $\eqn$. Then update the current value of $\eqn$ to $E \lub
  E'$.
  Once all these updates are (iteratively) applied we repeat the process at 
  step 1.
\end{enumerate}

\noindent
Termination is guaranteed since the abstract domain does not have infinitely
ascending chains as the number of jump destinations and the stack size are
finite.
\mycomment{This is the case of the programs that satisfy the
  constraints stated in Section~\ref{sec:evm-language}.}

% In the
% implementation, we apply several optimizations to improve the performance of the
% above process.
% We omit the details as they are not important for explaining the analysis.

\input eqsystem

\begin{example}

Figure~\ref{fig:eqsystem} shows the equations produced by
Definition~\ref{def:jumpeq} of the first and the last instruction of all blocks
shown in Figure~\ref{fig:cfg-ins}.
% %
The first instruction shown in the system is $\eq{64B}$, computed in
Example~\ref{ex:transfer}.
% %
Observe that the application of $\tau$ stores the jumping addresses in the
corresponding abstract states after \code{PUSH} instructions (see $\eq{660}$,
$\eq{66D}$, $\eq{6CF}$, $\eq{68C}$, \ldots). Such addresses will be used to
produce the equations at the \code{JUMP} or \code{JUMPI} instructions.
% %
In the case of \code{JUMP}, as the jump is unconditional, it only produces one
equation, e.g. $\eq{66E}$ consumes address \hdir{66F} to produce the input state
of $\eq{66F}$, or $\eq{6C2}$ produces the input abstract state for $\eq{6D1}$.
% %
\code{JUMPI} instructions produce two different equations: (1) one equation
which corresponds to the jumping address stored in the stack, e.g. equations
$\eq{6D0}$ and $\eq{66F}$ produced by the jumps of the equations $\eq{660}$ and
$\eq{66D}$ respectively; and (2) one equation which corresponds to the next
instruction, e.g. $\eq{661}$ and $\eq{66E}$ produced by $\eq{660}$ and
$\eq{66D}$, respectively.
% % %
Finally, another point to highlight occurs at equation $\eq{6D5}$: as we have
two possible jumping addresses in the stack of and both can be used by the
\code{JUMP} at the end of the block, we produce two inputs for the two
possible jumping addresses, $\eq{954}$ and $\eq{142}$, for capturing the two
possible branches from block \hdir{6D1} (see Figure~\ref{fig:cfg-ins}).
% %
\exampleend
\end{example}

\begin{theorem}[Soundness of the addresses equation system]
  \label{theorem:eqsystem}
  Let $P \equiv b_0,\dots,b_p$ be a program, $\eqn_1,\dots,\eqn_n$ the
  solution of the jumps equations system of $P$, and $pc$ the program
  counter of a jump instruction. Then for any execution trace $t$ of $P$, there
  exists $s \in dom(\eqn_{pc})$ such that
  $ \tuple{n, \sigma} \in \eqn_{pc}(s)$ and $\sigma(s_{n-1})$ contains
  all jump addresses that instruction $b_{pc}$ jumps to in $t$.
\end{theorem}

We follow the next steps to prove the soundness of this theorem:

\begin{enumerate}
\item We first define an \EVM collecting semantics for the operational
  semantics of Figure~\ref{fig:jumpsemantics}. Such collecting
  semantics gathers all transitions that can be produced by the
  execution of a program $P$.
\item We continue by defining the jumps-to property as a property of
  this collecting semantics.
\item Then we prove Lemma~\ref{lemma1} below that states that the
  least solution of the addresses equation system generated from the
  \EVM program as described in Definition~\ref{def:jumpeq} is a safe
  approximation of the \EVM collecting semantics w.r.t. the jumps-to
  property.
\item Finally, Theorem~\ref{theorem:eqsystem} trivially follows from
  Lemma~\ref{lemma1}.
\end{enumerate}

\newcommand{\colsemp}{\ensuremath{\mathcal{C}_P}}
\begin{definition}[\EVM collecting semantics]
  Given an \EVM program $P$, the \EVM collecting semantics operator
  $\colsemp$ is defined as follows:
\[
\colsemp (X) = \set{\tuple{S,S'} ~|~ \tuple{\_,S} \in X \wedge S \rrderivproof S'}
\]
The \EVM semantics is defined as $\xi_P = \bigcup_{n>0}
\colsemp^n(X_0)$,  where $X_0 \equiv \set{\tuple{0,\tuple{0,\sigma_{\emptyset}}}}$ is the initial configuration.

\end{definition}

\begin{definition}[jumps-to property]
  Let $P$ be an IR program, $\xi_P = \bigcup_{n>0} \colsemp^n(X_0)$,
  and $b$ an instruction at program point $pc$, then we say
  that $\xi_P \vDash_{pc} T$ ~if~
  $T = \set{ \tuple{n,\sigma} ~|~ \tuple{S,S'}\in\xi_P \wedge
    \tuple{pc,\tuple{n,\sigma}} \in S'}$.

\end{definition}

The following lemma states that the least solution of the constraint
equation system defined in Definition~\ref{eq:equation_system}
is a safe approximation of $\xi_P$:

\begin{lemma}
\label{lemma1}
  Let $P \equiv b_0,\dots, b_p$ be a program, $pc$ a program point and
  $\eqn_{0},\dots, \eqn_{p}$ the least solution of the constraints
  equation system of Definition~\ref{def:jumpeq}.
  The following holds:

  \begin{center}
    If $\xi_P \vDash_{pc} T$, then for all $\tuple{n,\sigma}\in
  T$, exists $s \in dom(\eqn_{pc})$ such that $\tuple{n,\sigma} \in\eqn_{pc}(s)$.
  \end{center}
\end{lemma}

\begin{proof}
We use $\eqn_{pc}^m$ to refer to the value obtained for $\eqn_{pc}$
after $m$ iterations of the algorithm for solving the equation system
depicted in Section~\ref{sec:from-evm-to-cfg}.
We say that $\eqn_{pc}$
\emph{covers} $\tuple{n,\sigma}$ in $\colsemp^m(X_0)$ at program point
$pc$ when this lemma holds for the result of computing
$\colsemp^m(X_0)$.
In order to prove this lemma, we can reason by induction on
the value of $m$, the length of the traces $S_0 \rrderivproof^m S_m$ considered
in $\colsemp^m(X_0)$.

\textbf{Base case}: if $m = 0$,
$S_0 = \tuple{0,\tuple{0, \sigma_{\emptyset}}}$ and the Lemma
trivially holds as 
$\tuple{0, \sigma_{\emptyset}} \in \eqn_0^0(\tuple{0,
  \sigma_{\emptyset}})$.

\textbf{Induction Hypothesis}: we assume Lemma~\ref{lemma1} holds for all
traces of length $m \geq 0$.

\textbf{Inductive Case}: Let us consider traces of length $m + 1$,
which are of the form $S_0 \rrderivproof^m S_m \rrderivproof S_{m+1}$. $S_m$ is
a program state of the form $S_m = \tuple{pc, \tuple{n, \sigma}}$.  We
can apply the induction hypothesis to $S_m$: there exists some
$s\in dom(\eqn_{pc}^m)$ such that
$ \tuple{n, \sigma} \in \eqn_{pc}^{m}(s) $. For extending the Lemma, we
reason for all possible rules in the simplified \EVM semantics
(Fig.~\ref{EVM_semantics}) we may apply from $S_m$ to $S_{m+1}$:

\begin{itemize}
\item{Rule (1)}: After executing a \code{JUMP} instruction $S_{m+1}$
  is of the form
  $\tuple{\sigma(s_{n-1}), \tuple{n-1,
      \sigma\backslash[s_{n-1}]}}$.
  In iteration $m+1$, the following set of equations corresponding to
  $b_{pc}$ is evaluated:
\[ 
\begin{array}{rcll}
\eqn_{\sigma(s_{n-1})} & \sqsupseteq & \textit{idmap}(\lambda(b_{pc},
\tuple{n', \sigma'})) & \forall s'\in dom(\eqn_{pc}), \tuple{n', \sigma'} \in \eqn_{pc}(s')
\end{array}
\]
where
$\textit{idmap}(\lambda(b_{pc}, \langle n', \sigma'
\rangle)) = \pi_{\bot}[ \tuple{n'-1,\sigma'\backslash[s_{n-1}]} \mapsto
\set{\tuple{n'-1,\sigma'\backslash[s_{n-1}]}}]$ (Case (4) in Fig.~\ref{fig:transfer}).
The induction hypothesis guarantees that there exists some $s'' \in
\eqn_{pc}^{m}$ such that 
$ \tuple{n, \sigma} \in \eqn_{pc}^{m}(s'') $, where $S_m = \tuple{pc,
  \tuple{n, \sigma}}$.  Therefore, at
Iteration $m+1$, the following must hold:
\[ 
\eqn_{\sigma(s_{n-1})}^{m+1} \sqsupseteq
\pi_{\bot}[ \tuple{n-1,\sigma\backslash[s_{n-1}]} \mapsto \set{\tuple{n-1,\sigma\backslash[s_{n-1}]}}]
\]
so $\tuple{n-1,\sigma\backslash[s_{n-1}]} \in
\eqn_{\sigma(s_{n-1})}^{m+1}(\tuple{n-1,\sigma\backslash[s_{n-1}]})$
and thus Lemma~\ref{lemma1} holds.

\item{Rules (2) and (3)}: After executing a \code{JUMPI} instruction,
  $S_{m+1}$ is either
  $\tuple{\sigma(s_{n-1}), \tuple{n-2, \sigma\backslash[s_{n-1},
      s_{n-2}]}}$
  or $\tuple{pc+size(b_{pc}), \tuple{n-2, \sigma\backslash[s_{n-1}, s_{n-2}]}}$,
  respectively.  In any of those cases the following sets of equations
  are evaluated:
\[
\begin{array}{rcll}
\eqn_{\sigma(s_{n-2})} & \sqsupseteq & \textit{idmap}(\lambda(\text{\code{JUMPI}},
\langle n', \sigma' \rangle)) & \forall s'\in dom(\eqn_{pc}), \tuple{n', \sigma'} \in \eqn_{pc}(s')
\\
\eqn_{pc+1} & \sqsupseteq & \textit{idmap}(\lambda(\text{\code{JUMPI}},
\langle n', \sigma' \rangle)) & \forall s'\in dom(\eqn_{pc}), \tuple{n', \sigma'} \in \eqn_{pc}(s')
\end{array}
\]
where\\
$\textit{idmap}(\lambda(b_{pc}, \langle n', \sigma' \rangle)) =
\pi_{\bot}[ \tuple{n'-2,\sigma'\backslash[s_{n-1},s_{n-2}]} \mapsto
\set{\tuple{n'-2,\sigma'\backslash[s_{n-1},s_{n-2}]}}]$
(Case (4) of the definition of the update function $\lambda$ in
Fig.~\ref{fig:transfer}). 
As in the previous case, the induction hypothesis guarantees that at
Iteration $m$ there exists $s'' \in
\eqn_{pc}^{m}$ such that 
$ \tuple{n, \sigma} \in \eqn_{pc}^{m}(s'') $. Therefore, in 
Iteration $m+1$, the following must hold:
\[ 
\begin{array}{rcl}
\eqn_{\sigma(s_{n-1})}^{m+1} & \sqsupseteq &
\pi_{\bot}[ \tuple{n-2,\sigma\backslash[s_{n-1},s_{n-2}]} \mapsto \set{\tuple{n-2,\sigma\backslash[s_{n-1},s_{n-2}]}}]
\\
\eqn_{pc+1}^{m+1} & \sqsupseteq &
\pi_{\bot}[ \tuple{n-2,\sigma\backslash[s_{n-1},s_{n-2}]} \mapsto \set{\tuple{n-2,\sigma\backslash[s_{n-1},s_{n-2}]}}]
\end{array}
\]
and thus Lemma~\ref{lemma1} holds for these cases as well.

\item{Rules (4) - (12)}: We will first consider the
  case in which any of these rules corresponds to an \EVM instruction
  followed by an instruction different from \code{JUMPDEST}.  All
  rules are similar, as they all use the set of equations generated by
  Case (4) in Definition~\ref{def:jumpeq}.  We will see Rule (4) in detail.

After executing a
  $\text{\code{PUSH}}x~v$ instruction, $S_{m+1}$ is
  $\tuple{pc+size(b_{pc}), \tuple{n+1, \sigma[s_{n} \mapsto
      \set{v}]}}$.  We have to prove that exists some $s
  \in dom(\eqn_{pc+size(b_{pc})})$ such that $\tuple{n+1, \sigma[s_{n} \mapsto
      \set{v}]} \in \eqn_{pc+size(b_{pc})}(s)$.
  The following set of equations is evaluated:
\begin{equation}
\label{eq:rule4-1}
\begin{array}{rcll}
\eqn_{pc+size(b_{pc})} & \sqsupseteq & \tau(\text{\code{PUSH}}x,
\eqn_{pc})
\end{array}
\end{equation}
By Definition~\ref{def:transfer}, $\tau (\text{\code{PUSH}}x,
\eqn_{pc}) = \pi'$, where
$\forall s' \in dom(\pi),$ $\pi'(s') = \lambda(\text{\code{PUSH}}x,
\eqn_{pc}(s'))$.   By the case (1) of the definition of the update
function $\lambda$, we have that:
\begin{equation}
  \label{eq:rule4-2}
\forall \tuple{n'', \sigma''} \in dom(\eqn_{pc}), \pi'(\tuple{n'',
  \sigma''}) = \tuple{n''+1, \sigma''[s_n \mapsto
  \set{v}]} 
\end{equation}

By the induction hypothesis, at Iteration $m$ there exists some $s \in
dom(\eqn^m_{pc})$ such that $\tuple{n,\sigma} \in \eqn^m_{pc}(s)$.
Therefore, by~\ref{eq:rule4-1} and~\ref{eq:rule4-2}, at Iteration
$m+1$ we have that the following holds:
\[
 s \in dom(\eqn^{m+1}_{pc+size(b_{pc})}) ~\text{and}~ 
 \tuple{n+1, \sigma[s_{n} \mapsto
   \set{v}]} \in \eqn_{pc+size(b_{pc})}(s)
\]
and thus Lemma~\ref{lemma1} holds for Rule (4).

\item{Rules (4) - (12)}, followed by a \code{JUMPDEST} instruction.
After executing any of these instructions, $S_{m+1}$ is
  $\tuple{pc+size(b_{pc}), \tuple{n''', \sigma'''}}$, where
  $\tuple{n''', \sigma'''}$ is obtained according to the rule from
  Figure~\ref{fig:jumpsemantics}.  
  We have to prove that exists some $s
  \in dom(\eqn_{pc+size(b_{pc})})$ such that $\tuple{n''', \sigma'''}
  \in \eqn_{pc+size(b_{pc})}(s)$. 
  The following set of equations is evaluated:
\begin{equation}
\label{eq:rule4-3}
\begin{array}{rcll}
\eqn_{pc+size(b_{pc})} & \sqsupseteq & \textit{idmap}(\lambda(b_{pc},
\langle n', \sigma' \rangle)) & \forall s'\in dom(\eqn_{pc}), \tuple{n', \sigma'} \in \eqn_{pc}(s')
\end{array}
\end{equation}
where
$\textit{idmap}(\lambda(b_{pc}, \langle n', \sigma' \rangle)) =
\pi_{\bot}[ \tuple{n'',\sigma'']} \mapsto
\set{\tuple{n'',\sigma''}}]$, where $n''$ and $\sigma`$ are obtained
according to the cases of the updating function detailed in
Figure~\ref{fig:transfer}.  We can see that
$\set{\tuple{n'',\sigma''}}]$ match the modification made to the state
$S_{m+1}$ by the corresponding rule of the semantics.  Therefore, at Iteration
there exists an $s = \set{\tuple{n'',\sigma''}}]$ such that
$\set{\tuple{n'',\sigma''}}] \in \eqn_{pc+size(b_{pc})}^{m+1}$, and
Lemma~\ref{lemma1} also holds.
\end{itemize}

When the algorithm stops Lemma~\ref{lemma1} holds, as for any
$pc$ $\eqn_{pc}^{m+1} \sqsupseteq \eqn_{pc}^{m}$ for each iteration of
the algorithm for solving the equation system of
Section~\ref{sec:from-evm-to-cfg}.
\proofend
\end{proof}

\section{Stack-Sensitive Control Flow Graph} 
% %
At this point, by means of the addresses equation system solution, we compute
the control flow graph of the program. In order to simplify the notation, given a
block $B_i$, we define the function $getId(i,\tuple{n,\sigma})$, which receives
the block identifier $i$ and an abstract stack $\tuple{n,\sigma}$ and returns a
unique identifier for the abstract stack $\tuple{n,\sigma} \in dom(\eqn_{i})$.
Similarly, $getStack(i,id)$ returns the abstract state $\tuple{n,\sigma}$ that
corresponds to the identifier $id$ of block $B_i$.
Besides, we define the function $getSize(pc,id)$ that, given a program point $pc
\in B_i$ and a unique identifier $id$ for $B_i$, returns the value $n'$ s.t.
$\tuple{n,\sigma} = getStack(i,id)$, and $\eqn_{pc}(\tuple{n,\sigma}) =
\tuple{n',\sigma'}$.
% of its corresponding abstract stack $\tuple{n,\sigma}$.

\begin{example}
\label{ex:getid}
Given the equation: 
{\small
\[
\eq{64B} \sqsupseteq 
\{
\underbrace{\tuple{7,\set{s_5 \mapsto \hdir{954}}}}_{1} \mapsto
\tuple{7,\set{s_5 \mapsto \hdir{954}}},
\underbrace{\tuple{3,\set{s_1 \mapsto \hdir{142}}}}_{2} \mapsto
\tuple{3,\set{s_1
\mapsto
\hdir{142}}}
\},
\]
}

\noindent
if we compute the functions $getId$ and $getSize$, we have that
$getId(\hdir{64B},\tuple{7,\set{s_5 \mapsto \hdir{954}}}) = 1$ and
$getId(\hdir{64B},\tuple{3,\set{s_1 \mapsto \hdir{142}}}) =
2$. Analogously, $getSize(\hdir{64B},1) = 7$ and
$getSize(\hdir{64B},2) = 3$.  \exampleend
\end{example} 

\begin{definition}[stack-sensitive control flow graph]
\label{def:cfg}
  Given an \EVM program
  $P$, its blocks $B_i \equiv b_i \dots b_j \in blocks(P)$
  % its
  % controlflow
  % graph
  % $CFG = \tuple{V,E}$
  and its \emph{flow} analysis results provided by a set of variables
  of the form
  % the function $\pi_{pp}$ (where $pp$ is the corresponding program
  % point),
  $\eqn_{pc}$ for all $pc \in P$, we define the \emph{control flow
    graph} of $P$ as a directed graph $\CFG=\tuple{V,E}$ with a set of
  vertices
% %
\[V = \{B_{i{:}id} ~|~ 
B_i \in blocks(P) \wedge 
\tuple{n,\sigma} \in dom(\eqn_i) \wedge
id = getId(i,\langle n, \sigma\rangle) \}\]
% %
 and a set of edges $E = E_{jump} \cup E_{next}$ such that:
\[
\begin{array}{rcll}
E_{jump} &=& \{B_{i{:}id} \to B_{d:id_2} ~|~ & b_j \in Jump ~\wedge 
    \\&&& 
    \langle n, \sigma \rangle \in dom(\eqn_{j}) \wedge id =
    getId(i,\tuple{n,\sigma}) ~\wedge~ 
    \\ &&& 
    \langle n', \sigma' \rangle \in \eqn_{j}(\langle n, \sigma \rangle) ~\wedge~
    d = \sigma'(s_{n'-1})  ~\wedge~ 
    \\ &&& 
	\tuple{n'',\sigma''} =\lambda(b_j,\langle n',\sigma' \rangle) \wedge 
    id_2 = getId(d,\tuple{n'',\sigma''}) 
             ~\}
             %\wedge n'=n-1 \wedge \sigma' = \sigma\backslash[s_{n-1}]\}\\
             
\\

E_{next} &=& \{B_{i:id} \to B_{d:id_2} ~|~ & 
				b_j \neq \text{\code{JUMP}} ~\wedge b_j \not\in End ~\wedge
	\\ &&& 
    \langle n, \sigma \rangle \in dom(\eqn_{j}) \wedge id =
    getId(i,\tuple{n,\sigma}) ~\wedge~ 
     \\ &&& \langle n', \sigma' \rangle \in \eqn_{j}(\langle n,
     \sigma \rangle) \wedge d = j+size(b_j) ~\wedge 
     \\ &&&
	\tuple{n'',\sigma''} =\lambda(b_j,\langle n',\sigma' \rangle) \wedge 
    id_2 = getId(d,\tuple{n'',\sigma''}) 
             
\end{array}
\]

\end{definition}

The first relevant point of the control flow graph (CFG) we produce is
that, for producing the set of vertices $V$, we replicate each block
for each different stack state that could be used for invoking it.
%(gray nodes in Figure~\ref{fig:cfg-ins} are replicated in the CFG).
% %
Analogously, the different entry stack states are also used to produce
different edges depending on its corresponding replicated blocks. Note
that the definition distinguishes between two kinds of edges. (1)
edges produced by \code{JUMP} or \code{JUMPI} instructions at the end
of the blocks, whose destination is taken from the values stored in
the stack states of the instruction before the jump with
$d = \sigma'(s_{n'-1})$; and (2) edges produced by continuations to
the next instruction, whose destination is computed with
$d = j + size(b_j)$.
% %
In both kinds of edges, as we could have replicated blocks, we apply
function $\lambda$ and get the id of the resulting state to compute
the $id$ of the destination:
$\tuple{n'',\sigma''} =\lambda(b_j,\langle n',\sigma' \rangle) \wedge
id_2 = getId(d,\tuple{n'',\sigma''})$.

\begin{example}
\label{ex:cfg}

Considering the blocks shown in Figure~\ref{fig:cfg-ins} and the
equations shown at Figure~\ref{fig:eqsystem}, the CFG of the program
includes non-replicated nodes for those blocks that only receive one
possible stack state (white nodes in Figure~\ref{fig:cfg-ins}).
However, the nodes that could be reached by two different stack states
(gray nodes in Figure~\ref{fig:cfg-ins}) will be replicated in the
CFG:
\[
\begin{array}{rcl}
V & = \{
B_\hdir{941},
B_\hdir{123},
B_\hdir{954},
B_\hdir{142}, &
B_\hdir{64B{:}1},
B_\hdir{653{:}1}, 
B_\hdir{661{:}1},
B_\hdir{66F{:}1},
B_\hdir{6C3{:}1},
B_\hdir{66E{:}1},
B_\hdir{690{:}1},
B_\hdir{683{:}1},
B_\hdir{691{:}1},
B_\hdir{6D0{:}1},
B_\hdir{6D1{:}1},
\\ &&
B_\hdir{64B{:}2},
B_\hdir{653{:}2}, 
B_\hdir{661{:}2},
B_\hdir{66F{:}2},
B_\hdir{6C3{:}2},
B_\hdir{66E{:}2},
B_\hdir{690{:}2},
B_\hdir{683{:}2},
B_\hdir{691{:}2},
B_\hdir{6D0{:}2},
B_\hdir{6D1{:}2}

\}
\end{array}
\]

\noindent
Analogously, our CFG replicates the edges according to the nodes
replicated (solid and dashed edges in Figure~\ref{fig:cfg-ins}):

\begin{tabular}{rcp{14.5cm}l}
E  &=&
\{
$B_\hdir{941} \rightarrow B_\hdir{64B{:}1},$
$B_\hdir{64B{:}1} \rightarrow B_\hdir{653{:}1},$
$B_\hdir{653{:}1} \rightarrow B_\hdir{661{:}1},$
$B_\hdir{661{:}1} \rightarrow B_\hdir{66F{:}1},$
$B_\hdir{66F{:}1} \rightarrow B_\hdir{6C3{:}1},$
$B_\hdir{6C3{:}1} \rightarrow B_\hdir{653{:}1},$
$B_\hdir{66D{:}1} \rightarrow B_\hdir{66E{:}1},$
$B_\hdir{66F{:}1} \rightarrow B_\hdir{690{:}1},$
$B_\hdir{66F{:}1} \rightarrow B_\hdir{683{:}1},$
$B_\hdir{683{:}1} \rightarrow B_\hdir{691{:}1},$
$B_\hdir{691{:}1} \rightarrow B_\hdir{6D1{:}1},$
$B_\hdir{6D1{:}1} \rightarrow B_\hdir{954},$
$B_\hdir{123} \dashrightarrow B_\hdir{64B{:}2},$
$B_\hdir{64B{:}2} \dashrightarrow B_\hdir{653{:}2},$
$B_\hdir{653{:}2} \dashrightarrow B_\hdir{661{:}2},$
$B_\hdir{661{:}2} \dashrightarrow B_\hdir{66F{:}2},$
$B_\hdir{66F{:}2} \dashrightarrow B_\hdir{6C3{:}2},$
$B_\hdir{6C3{:}2} \dashrightarrow B_\hdir{653{:}2},$
$B_\hdir{66D{:}2} \dashrightarrow B_\hdir{66E{:}2},$
$B_\hdir{66F{:}2} \dashrightarrow B_\hdir{690{:}2},$
$B_\hdir{66F{:}2} \dashrightarrow B_\hdir{683{:}2},$
$B_\hdir{683{:}2} \dashrightarrow B_\hdir{691{:}2},$
$B_\hdir{691{:}2} \dashrightarrow B_\hdir{6D1{:}2},$
$B_\hdir{6D1{:}2} \dashrightarrow B_\hdir{142}$
\}
\end{tabular}
\medskip

\noindent Note that, in Figure~\ref{fig:cfg-ins}, we distinguish
dashed and solid edges just to remark that we could have two possible
execution paths, that is, if the call to \code{findWinner} comes from
block $B_\hdir{941}$, it will return to block $B_\hdir{954}$ and, if
the execution comes from a public invocation, \ie block
$B_\hdir{123}$, it will return to block $B_\hdir{142}$.  \exampleend
\end{example}

\begin{theorem}[Soundness of the stack-sensitive control flow graph]
  \label{theorem:scfg}
  Let $P$ be an EVM program. If a stack-sensitive control flow graph
  \CFG\ can be generated, then for any execution trace $t$ of $P$
  there exists a directed walk that visits, in the same order, nodes
  in the \CFG\ that correspond to replicas of the blocks executed in
  $t$.
\end{theorem}

\begin{proof}
  We prove this theorem reasoning by induction on the value of $m$,
  the length of the trace $t \equiv S_0 \rrderivproof^m S_m$. We will
  assume that a directed walk of the \CFG\ is of the form
  $B_{0:0}\cdot\ldots\cdot B_{n:id_n}$, where $B_{0:0}$ is a replica
  of the block that contains the first instruction in the program
  $b_0$.

\textbf{Base case}: if $m = 0$,
$S_0 = \tuple{0,\tuple{0, \sigma_{\emptyset}}}$ and the Lemma
trivially holds as $b_0$ is the first instruction of block $B_0$.

\textbf{Induction Hypothesis}: we assume Theorem~\ref{theorem:scfg}
holds for all traces of length $m \geq 0$.

\textbf{Inductive Case}: Let us consider a trace of length $m + 1$,
$t \equiv S_0 \rrderivproof^m S_m \rrderivproof S_{m+1}$. $S_m$ is a
program state of the form $S_m = \tuple{pc, \tuple{n_m, \sigma_m}}$.  We can apply the
induction hypothesis to $S_m$: there exists a directed walk in \CFG,
$w \equiv B_{0:0}\cdot \ldots \cdot B_{j:id}$ that visits nodes
corresponding 
to replicas of the blocks executed in $t$ in the same
order, and $b_{pc}\in B_{j:id}$. There may be two cases:

\begin{itemize}
\item [a)] Instruction $b_{pc}$ is not the last instruction in
  $B_{j:id}$. By Definition~\ref{def:block}, $b_{pc+size(b_{pc})}$ is
  also in $B_{j:id}$, and $b_{pc}\not\in Jump$. The applicable rules
  of the semantics of Figure~\ref{fig:jumpsemantics} are Rules
  \textsc{(4)} to \textsc{(12)}. In all cases,
  $S_{m+1} = \tuple{pc+size(b_{pc}), \tuple{n_{m+1}, \sigma_{m+1}}}$ and 
  Theorem~\ref{theorem:scfg} holds, since the same directed walk $w$
  already visits a replica of the node that contains the instruction
  executed in $S_{m+1}$.

\item [b)] Instruction $b_{pc}$ is the last instruction in block
  $B_{j:id}$.  We reason on all possible instructions that can be the
  last instruction of a block:

  \begin{itemize}
  \item $b_{pc}\not\in Jump$. This case is the result of the
    application of Rules \textsc{(4-12)} of the Semantics in
    Figure~\ref{EVM_semantics}. Therefore, $S_{m+1}$ is of the form
    $S_{m+1} = \tuple{d, \tuple{n_{m+1}, \sigma_{m+1}}}$, where
    $d = b_{pc + size(b_{pc})}$.
    Lemma~\ref{lemma1} guarantees that there exists a stack state
    $\tuple{n,\sigma}$ such that
    $\tuple{n_{m}, \sigma_{m}} \in \eqn_{pc}(\tuple{n,\sigma})$. By
    Definition~\ref{def:cfg} there is an edge in $E_{next}$ of the
    form $B_{j:id} \to B_{d:id_2}$ where $d = pc+size(b_{pc})$ for
    each element in $\eqn_{pc}(\tuple{n,\sigma})$. Therefore, there
    exists a 
    directed walk $w' = w\cdot B_{d:id_2}$ that visits nodes in \CFG\
    corresponding to replicas of the blocks executed in $t$, and
    Theorem~\ref{theorem:scfg} holds.

  \item $b_{pc}\in Jump$. This case is the result of the
    application of Rules~\textsc{(1-3)} of the Semantics in
    Figure~\ref{EVM_semantics}. 

    The application of Rules~\textsc{(1-2)} corresponds to a jump in
    the code, and $S_{m+1}$ is of the form
    $S_{m+1} = \tuple{d_{m+1}, \tuple{n_{m+1}, \sigma_{m+1}}}$, where
    $d_{m+1} = \sigma_m(s_{n_m-1})$. 
    Lemma~\ref{lemma1} guarantees that there exists a stack state
    $\tuple{n,\sigma}$ such that
    $\tuple{n_{m}, \sigma_{m}} \in \eqn_{pc}(\tuple{n,\sigma})$. By
    Definition~\ref{def:cfg}, $E_{Jump}$ contains an edge
    $B_{j{:}id} \to B_{d:id_2}$ for each element in
    $\eqn_{pc}(\tuple{n,\sigma})$, such that
    $d = \sigma_{m}(s_{n_m-1})$ and $id_2$ is the replica identifier
    of block $d$ corresponding to
    $\lambda(b_{pc},\tuple{n_{m}, \sigma_{m}})$. Therefore,
    $d = d_{m+1}$, and the directed walk $w' = w\cdot B_{d:id_2}$
    visits nodes in \CFG\ corresponding to replicas of the blocks
    executed in $t$, so Theorem~\ref{theorem:scfg} holds.

    The application of Rule~\textsc{(3)} corresponds to a \code{JUMPI}
    instruction in the code that does not jump to its destination
    address. This case is equal to the previous case in which
    $b_{pc}\not\in Jump$. 
  \end{itemize}

\end{itemize}

  \proofend
\end{proof}

%% file: transfer.tex
% ;; -*- coding: iso-latin-1; TeX-PDF-mode: t; TeX-master: "main" -*-%

\begin{figure}

\noindent
\begin{center}
{\scriptsize
{\renewcommand{\arraystretch}{1.4}
\begin{tabular}{r|l|ll|}
\cline{2-4}
& 
\multicolumn{1}{c|}{$b^{\delta,\alpha}$}        & 
\multicolumn{2}{c|}{$\lambda(b, \langle n,\sigma\rangle)$} \\[0.05cm]
\cline{2-4}
\multirow{2}{*}{{\scriptsize (1)}} & 
\multirow{2}{*}{\code{PUSH}$x$~$v$}   & 
$\langle n+1, \sigma[s_n \mapsto \{v\}] \rangle$
& 
\textit{when} $v \in \mathcal{J}$
\\
 &
 &
$\langle n+1, \sigma \rangle$
& 
\textit{when} $v \not\in \mathcal{J}$
\\
\cline{2-4}
\multirow{2}{*}{{\scriptsize (2)}} & 
\multirow{2}{*}{\code{DUP}x}   & 
$\langle n+1, \sigma \rangle$ & 
\textit{when} $s_{n-x} \not\in dom(\sigma)$
\\
 & 
 & 
$\langle n+1, \sigma[s_n \mapsto \sigma(s_{n-x})] \rangle$
& \textit{when} $s_{n-x} \in dom(\sigma)$
\\
\cline{2-4}
\multirow{4}{*}{{\scriptsize (3)}} & 
\multirow{5}{*}{\code{SWAP}x}   & 
$\langle n, \sigma \rangle$ & 
\textit{when} $s_{n-1} \not\in dom(\sigma) \wedge s_{n{-}x{-}1} \not\in dom(\sigma)$
\\
&
&
$\langle n, \sigma[s_{n-x-1} \mapsto \sigma(s_{n-1}), s_{n-1} \mapsto \sigma(s_{n-x-1})]
    \rangle$ & \textit{when} $s_{n-1} \in dom(\sigma) \wedge s_{n{-}x{-}1} \in dom(\sigma)$
\\
&
&
$\langle n, \sigma[s_{n{-}1} \mapsto \sigma(s_{n{-}x{-}1})]\backslash\sigma[s_{n{-}x{-}1}] \rangle$ & 
\textit{when} $s_{n{-}1} \not\in dom(\sigma) \wedge s_{n{-}x{-}1} \in dom(\sigma)$
\\
& 
&
$\langle n, \sigma[s_{n{-}x{-}1} \mapsto \sigma(s_{n{-}1})]\backslash\sigma[s_{n{-}1}] \rangle$
& \textit{when} $s_{n{-}1} \in dom(\sigma) \wedge s_{n{-}x{-}1} \not\in dom(\sigma)$
\\
\cline{2-4}
{\scriptsize (4)} & 
\textit{otherwise}   & 
$\langle n-\delta+\alpha, \sigma\backslash[s_{n-1},\dots,s_{n-\delta}] \rangle$ & 
\\
\cline{2-4}

\end{tabular}
}
}
\end{center}
\caption{Updating function}
\label{fig:transfer}
\end{figure}

%% file: eqsystem.tex
\begin{figure}[t]
\noindent
{\small
\[
\begin{array}{rclll}

\hspace{-0,2cm}\anno{A}
\eq{64B} & 
\sqsupseteq & 
\{
\tuple{7,\set{s_5 \mapsto \hdir{954}}} \mapsto 
    \tuple{7,\set{s_5 \mapsto \hdir{954}}}, & 
\tuple{3,\set{s_1 \mapsto \hdir{142}}} \mapsto 
    \tuple{3,\set{s_1 \mapsto \hdir{142}}} 
&\}
\\
& \vdots 
\\
\eq{652} & 
\sqsupseteq & 
\{
\tuple{7,\set{s_5 \mapsto \hdir{954}}} \mapsto 
    \tuple{9,\set{s_5 \mapsto \hdir{954}}}, &
\tuple{3,\set{s_1 \mapsto \hdir{142}}} \mapsto 
    \tuple{5,\set{s_1 \mapsto \hdir{142}}} 
&\}
\\
\hspace{-0,2cm}\anno{A}
\eq{653} & 
\sqsupseteq & 
\{
\tuple{9,\set{s_5 \mapsto \hdir{954}}} \mapsto 
    \tuple{9,\set{s_5 \mapsto \hdir{954}}}, &
\tuple{5,\set{s_1 \mapsto \hdir{142}}} \mapsto 
     \tuple{5,\set{s_1 \mapsto \hdir{142}}} 
&\}
\\ 
& \vdots \\
\eq{660} & 
\sqsupseteq & 
\{
\tuple{9,\set{s_5 \mapsto \hdir{954}}} \mapsto 
   \tuple{11,\set{s_5 \mapsto \hdir{954}, s_{10} \mapsto \hdir{6D0}}}, &
\tuple{5,\set{s_1 \mapsto \hdir{142}}} \mapsto 
    \tuple{7,\set{s_1 \mapsto \hdir{142}, s_6 \mapsto \hdir{6D0}}}
&\}
\\
\hspace{-0,2cm}\anno{A}
\eq{661} & 
\sqsupseteq & 
\{
\tuple{9,\set{s_5 \mapsto \hdir{954}}} \mapsto 
     \tuple{9,\set{s_5 \mapsto \hdir{954}}}, & 
\tuple{5,\set{s_1 \mapsto \hdir{142}}} 
     \mapsto \tuple{5,\set{s_1 \mapsto \hdir{142}}} 
&\}
\\
\eq{6D0} & 
\sqsupseteq & 
\{
\tuple{9,\set{s_5 \mapsto \hdir{954}}} \mapsto 
     \tuple{9,\set{s_5 \mapsto \hdir{954}}}, & 
\tuple{5,\set{s_1 \mapsto \hdir{142}}} \mapsto 
     \tuple{5,\set{s_1 \mapsto \hdir{142}}} 
&\}
\\
& \vdots \\
\eq{66D} & 
\sqsupseteq & 
\{
\tuple{9,\set{s_5 \mapsto \hdir{954}}} \mapsto 
     \tuple{13,\set{s_5 \mapsto \hdir{954},s_{12} \mapsto \hdir{66F}}}, &
\tuple{5,\set{s_1 \mapsto \hdir{142}}} \mapsto 
     \tuple{9,\set{s_1\mapsto \hdir{142},s_{10} \mapsto \hdir{66F}}}
&\}
\\
\eq{66E} & 
\sqsupseteq & 
\{
\tuple{11,\set{s_5 \mapsto \hdir{954}}} \mapsto 
     \tuple{11,\set{s_5 \mapsto \hdir{954}}}, &
\tuple{7,\set{s_1\mapsto \hdir{142}}} \mapsto 
     \tuple{7,\set{s_1\mapsto \hdir{142}}}
&\}
\\
\hspace{-0,2cm}\anno{A}
\eq{66F} & 
\sqsupseteq & 
\{
\tuple{11,\set{s_5 \mapsto \hdir{954}}} \mapsto 
     \tuple{11,\set{s_5 \mapsto \hdir{954}}}, &
\tuple{7,\set{s_1\mapsto \hdir{142}}} \mapsto 
     \tuple{7,\set{s_1\mapsto \hdir{142}}}
&\}
\\

& \vdots \\
\eq{682} & 
\sqsupseteq & 
\{
\tuple{11,\set{s_5 \mapsto \hdir{954}}} \mapsto 
    \tuple{11,\{s_5 \mapsto \hdir{954},s_{10} \mapsto \hdir{6C3}}, &
\tuple{7,\set{s_1\mapsto \hdir{142}}} \mapsto 
    \tuple{7,\{s_1 \mapsto \hdir{142}, s_6 \mapsto \hdir{6C3}\}}
&\}

\\
\hspace{-0,2cm}\anno{A}
\eq{6C3} & 
\sqsupseteq &
\{ 
\tuple{9,\set{s_5 \mapsto \hdir{954}}} \mapsto 
    \tuple{9,\set{s_5 \mapsto \hdir{954}}}, &
\tuple{5,\set{s_1\mapsto \hdir{142}}} \mapsto 
    \tuple{5,\set{s_1 \mapsto \hdir{142}}}
&\}
\\
& \vdots 
\\
\eq{6CF} & 
\sqsupseteq & 
\{
\tuple{9,\set{s_5 \mapsto \hdir{954}}} \mapsto 
     \tuple{10,\{s_5 \mapsto \hdir{954},s_{9} \mapsto \hdir{653}}, &
\tuple{5,\set{s_1\mapsto \hdir{142}}} \mapsto 
     \tuple{6,\{s_1 \mapsto \hdir{142}, s_5 \mapsto \hdir{653}\}}
&\}

\\[0.2cm]

\hspace{-0,2cm}\anno{A}
\eq{683} & 
\sqsupseteq & 
\{
\tuple{9,\set{s_5 \mapsto \hdir{954}}} \mapsto 
    \tuple{9,\set{s_5 \mapsto \hdir{954}}}, &
\tuple{5,\set{s_1\mapsto \hdir{142}}} \mapsto 
    \tuple{5,\set{s_1 \mapsto \hdir{142}}}
&\}
\\ 
& \vdots \\
\eq{68F} & 
\sqsupseteq & 
\{
\tuple{9,\set{s_5 \mapsto \hdir{954}}} \mapsto 
    \tuple{13,\{s_5 \mapsto \hdir{954},s_{12} \mapsto \hdir{691}}, &
\tuple{5,\set{s_1\mapsto \hdir{142}}} \mapsto 
    \tuple{9,\{s_1 \mapsto \hdir{142}, s_{8} \mapsto \hdir{691}\}}
&\}
\\
\hspace{-0,2cm}\anno{A}
\eq{690} & 
\sqsupseteq & 
\{
\tuple{11,\set{s_5 \mapsto \hdir{954}}} \mapsto 
    \tuple{11,\{s_5 \mapsto \hdir{954}}, &
\tuple{7,\set{s_1\mapsto \hdir{142}}} \mapsto 
    \tuple{7,\{s_1 \mapsto \hdir{142}\}}
&\}
\\
\hspace{-0,2cm}\anno{A}
\eq{691} & 
\sqsupseteq & 
\{
\tuple{11,\set{s_5 \mapsto \hdir{954}}} \mapsto 
    \tuple{11,\{s_5 \mapsto \hdir{954}}, &
\tuple{7,\set{s_1\mapsto \hdir{142}}} \mapsto 
    \tuple{7,\{s_1 \mapsto \hdir{142}\}}
&\}

\\
& \vdots
\\
\eq{6C2} & 
\sqsupseteq & 
\{
\tuple{11,\set{s_5 \mapsto \hdir{954}}} \mapsto 
      \tuple{10,\{s_5 \mapsto \hdir{954},s_{9} \mapsto \hdir{6D1}}, &
\tuple{7,\set{s_1\mapsto \hdir{142}}} \mapsto      
      \tuple{6,\{s_1 \mapsto \hdir{142}, s_{5} \mapsto \hdir{6D1}\}}
&\}
\\
\hspace{-0,2cm}\anno{A}
\eq{6D1} & 
\sqsupseteq & 
\{
\tuple{9,\{s_5 \mapsto \hdir{954}\}} \mapsto 
    \tuple{9,\{s_5 \mapsto \hdir{954}\}}, &
\tuple{5,\{s_1 \mapsto \hdir{142}\}} \mapsto 
    \tuple{5,\{s_1 \mapsto \hdir{142}\}}
&\}

\\
\eq{6D2} & 
\sqsupseteq & 
\{
\tuple{9,\{s_5 \mapsto \hdir{954}\}} \mapsto 
    \tuple{8,\{s_5 \mapsto \hdir{954}\}}, &
\tuple{5,\{s_1 \mapsto \hdir{142}\}} \mapsto     
    \tuple{4,\{s_1 \mapsto \hdir{142}\}}
&\}
\\
\eq{6D3} & 
\sqsupseteq & 
\{
\tuple{9,\{s_5 \mapsto \hdir{954}\}} \mapsto 
    \tuple{8,\{s_7 \mapsto \hdir{954}\}}, &
\tuple{5,\{s_1 \mapsto \hdir{142}\}} \mapsto 
    \tuple{4,\{s_3 \mapsto \hdir{142}\}}
&\}
\\
\eq{6D4} & 
\sqsupseteq & 
\{
\tuple{9,\{s_5 \mapsto \hdir{954}\}} \mapsto 
    \tuple{8,\{s_6 \mapsto \hdir{954}\}}, &
\tuple{5,\{s_1 \mapsto \hdir{142}\}} \mapsto 
    \tuple{4,\{s_2 \mapsto \hdir{142}\}}
&\}
\\
\eq{6D5} & 
\sqsupseteq & 
\{
\tuple{9,\{s_5 \mapsto \hdir{954}\}} \mapsto 
    \tuple{7,\{s_6 \mapsto \hdir{954}\}}, &
\tuple{5,\{s_1 \mapsto \hdir{142}\}} \mapsto 
    \tuple{3,\{s_2 \mapsto \hdir{142}\}}
&\}
\\
\hspace{-0,2cm}\anno{B}
\eq{954} & 
\sqsupseteq & 
\{
\tuple{6,\{\}} \mapsto \tuple{6,\{\}}
&&\}
\\
& \vdots 
\\
\hspace{-0,2cm}\anno{B}
\eq{142} & 
\sqsupseteq & 
\{
\tuple{2,\{\}} \mapsto \tuple{2,\{\}}
&&\}
\\
& \vdots 
\end{array}
\]
}
\caption{Jumps equations system of \code{__callback} function}
\label{fig:eqsystem}
\end{figure}